%% file: main.tex
\documentclass[letterpaper,twocolumn,10pt]{article}
\usepackage{usenix}

\input{macros}

\begin{document}

\title{\Large \bf BFT-PoLoc: A Byzantine Fortified Trigonometric Proof of Location Protocol using Internet Delays}
\author{
{\rm Peiyao Sheng}\\
Witness Chain
\and
{\rm Vishal Sevani}\\
Witness Chain
\and
 {\rm Ranvir Rana}\\
 Witness Chain
\and
 {\rm Himanshu Tyagi}\\
 Witness Chain
 \and
 {\rm Pramod Viswanath\thanks{Corresponding author: \href{mailto:pramod.viswanath@kaleidoscope-blockchain.com}{viswanath.pramod@gmail.com} }}\\
 Witness Chain
} 

\maketitle

\input{abstract}
\input{01_introduction}

\input{02_related_work}

\input{03_model}

\input{04_protocol}
\input{05_calibration}

\input{06_evaluation}

\input{07_conclusion}

\section*{Acknowledgment}
This work was conducted when all authors were working at Witness Chain (\href{https://witnesschain.com/}{https://witnesschain.com/}). We thank Sagnik Bhattacharya, Mustafa Doger and Advaith Suresh for their help with the discussions and implementation.

\bibliographystyle{IEEEtranS}
\bibliography{pol}

\end{document}

%% file: macros.tex
\usepackage{graphicx}
\usepackage{amsfonts}
\usepackage{booktabs} 
\usepackage{tabularx}
\usepackage{makecell}
\usepackage{verbatim}
\usepackage{amsgen,amsmath,amstext,amsbsy,amsopn,amsthm}
\usepackage{amssymb}
\usepackage{cancel}
\usepackage{xspace}
\usepackage[most]{tcolorbox}
\usepackage{algorithm}
\usepackage{enumitem}
\usepackage[noend]{algpseudocode}
\usepackage{algorithmicx}
\usepackage{hyperref}
\usepackage{multicol}
\usepackage{multirow}
\usepackage[labelformat=simple]{subcaption}
\usepackage{caption}
\usepackage{upquote}
\usepackage{soul}
\usepackage[normalem]{ulem}
\usepackage{cancel}
 \usepackage{bbm}
\usepackage{threeparttable}
\usepackage{float}
\usepackage{pifont}
\usepackage{fancyhdr}
\usepackage{mathtools}
\usepackage{tablefootnote}
\usepackage{tikz}
\usepackage{textcomp}
\newtheorem{theorem}{Theorem}

%% file: abstract.tex
\begin{abstract}
Internet platforms depend on accurately determining the geographical locations of online users to deliver targeted services (e.g., advertising). The advent of decentralized platforms (blockchains) emphasizes the importance of geographically distributed nodes, making the validation of locations more crucial. In these decentralized settings, mutually non-trusting participants  need to {\em prove} their locations to each other. The incentives for claiming desired location include decentralization properties (validators of a blockchain), explicit rewards for improving coverage (physical infrastructure blockchains) and regulatory compliance -- and entice participants towards prevaricating their true location malicious via VPNs, tampering with internet delays, or compromising other parties (challengers) to misrepresent their location. Traditional delay-based geolocation methods focus on reducing the noise in measurements and are very vulnerable to wilful divergences from prescribed protocol. 

In this paper we use Internet delay measurements to securely prove the location of IP addresses while being immune to a large fraction of Byzantine actions. Our core methods are to endow Internet telemetry tools (e.g., ping) with cryptographic primitives (signatures and hash functions) together with Byzantine resistant data inferences subject to Euclidean geometric constraints. We introduce two new networking protocols, robust against Byzantine actions: Proof of Internet Geometry (PoIG) converts delay measurements into precise distance estimates across the Internet; Proof of Location (PoLoc) enables accurate and efficient multilateration of a specific IP address. The key algorithmic innovations are in conducting ``Byzantine fortified trigonometry" (BFT) inferences of data, endowing low rank matrix completion methods with Byzantine resistance.  

We implemented both PoIG and PoLoc protocols and integrated them into a fully functional delay-based geolocation service, with global coverage (with specific focus on US) and integrated into a major blockchain (Ethereum). In a baseline evaluation of our protocols, we demonstrate significant improvements in the accuracy and robustness of location verification. In particular, location is identified within 100 km for a large fraction of the area. 
Under Byzantine distance inflation attack (with majority honest challengers), the accuracy of PoIG remains above 95\%. 
The precision and Byzantine resistance improves with a greater diversity of the challengers’ ISPs and an increased number of challengers across different directions.


\end{abstract}

%% file: 01_introduction.tex
\section{Introduction}
\label{sec:introduction}
\paragraph{Internet and geolocation} In today's digital era, the geographical positioning of online participants is crucial for the effective operation of Internet services. In the era of Web 2.0, location data is fundamental to the provision of targeted advertising and customized services, serving as a key component of various business models. With the advent of Web 3.0, grounded in blockchain technologies, the significance of geolocation has become even more pronounced. The core principle of decentralization in public blockchain networks\cite{ethereum,bitcoin,solana} emphasizes the need for {\em geographically} dispersed nodes to prevent censorship, highlighting the importance of accurately gauging the locations of these nodes. Moreover, geolocation plays a vital role in a wide range of distributed services, including platforms for distributed file storage \cite{filecoin, sia}, VPN services \cite{mysterium}, and distributed wireless networks \cite{helium}. These critical services require that hosts are physically located where they claim, ensuring data redundancy and resilience in file storage services, and the reliability of VPN services for privacy and accessing geo-restricted content. 

\begin{figure*}
\centering
\begin{subfigure}{.45\textwidth}
  \centering
  \includegraphics[width=\linewidth]{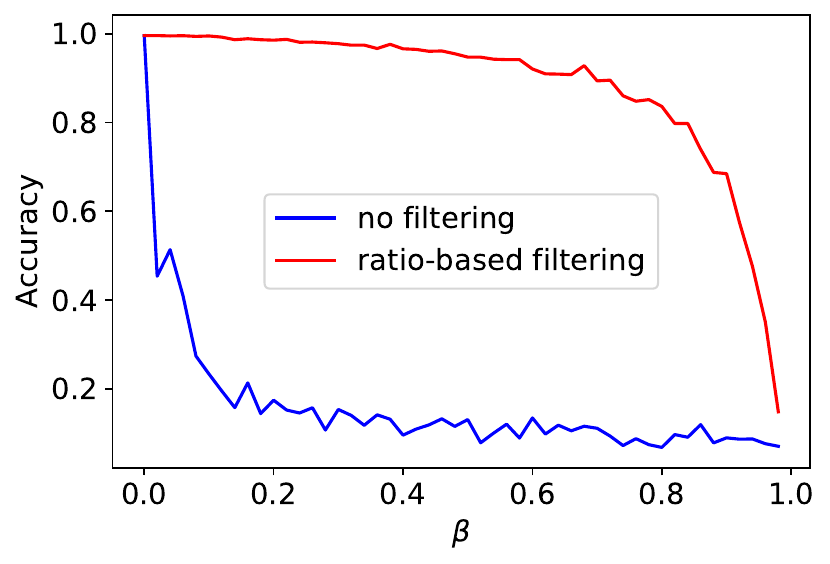}
  \caption{Accuracy-adversarial fraction ($\beta$) curve}
  \label{fig:acc}
\end{subfigure}\qquad  \qquad
\begin{subfigure}{.4\textwidth}
  \centering
  \includegraphics[width=\linewidth]{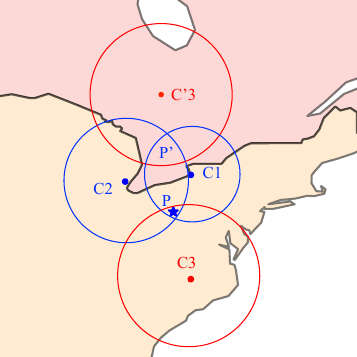}
  \caption{Feasible region of Waldo}
  \label{fig:circles}
\end{subfigure}
\caption{(a) The accuracy of delay-to-distance mapping is significantly improved under Byzantine distance inflation attack when employing robust ratio-based filtering. Here we assume delay is a linear function of distance with exponential noises, and the accuracy is the difference of MLE estimator to the coefficient. (b) An example of PoLoc protocol with three challengers located at $C_1, C_2, C_3$ shows that traditional multilateration method can be manipulated by Byzantine challenger who claims a false location $C'_3$ and outputs the incorrect region in Mexico for $P'$ instead of the region in the US which contains Waldo's true location $P$. Our PoLoc protocol outputs the largest uncertainty in all directions including $P'P$, providing robustness against Byzantine challengers. }
\label{fig:main-results}
\end{figure*}

\paragraph{ Delay based geolocation} The Internet's dynamic nature necessitates a robust framework for geolocation, one that can handle the inherent noise in measurement data. The key is to tie Internet delay measurements between pairs of addresses into geographic distances, relying on the premise that the propagation speed of EM/optical signals over wireless or wireline (including fiber optic)  cables is relatively constant (e.g., a constant fraction of the speed of light).
The challenge in this premise is that it conflicts  with the complex reality that Internet delays are not directly proportional to physical distances. This discrepancy arises from factors such as the varied speed of signal propagation across different media and the circuitous paths data packets traverse, shaped by geopolitical, commercial, and infrastructural considerations. 

\paragraph{Two stages of geolocation} To address this challenge, a standard approach is to divide the geolocation framework consists into two distinct stages. In the first stage, the {\em Internet geometry} is learnt: this refers to the delay to distance mapping (without regard to the angle) for a given IP address -- using the {\em known} (and trusted) locations of certain addresses (so-called ``landmarks''). This mapping is  navigating the complexities of Internet topology and latency influences via learning   network coordinate systems~\cite{dabek2004vivaldi,mao2004modeling},   delay-to-distance functions~\cite{gueye2004constraint, padamanabban2001determining,kohls2022verloc} and utilization of data-driven learning methods~\cite{lee2016ip,hong2023cheap}. 
In the second ``proof of location'' stage, a given IP address is geolocated via ``triangulation'', using 
  both the Internet geometry and delay measurements via trusted and location-aware landmarks~\cite{gueye2004constraint, padamanabban2001determining, laki2011spotter, wong2007octant, hong2023cheap, lee2016ip, maram2021goat}. The goal is typically to find the smallest feasible region or the most possible location the given address might be in face of measurement noises.

\paragraph{Byzantine fault tolerant geolocation} The above approaches largely deal with the uncertainty and noise in converting Internet delay measurements into distances and subsequently using these distances to determine locations. However, these methods are readily circumvented by the {\em adversarial} behavior, including the usage of VPNs and deliberate delay manipulations~\cite{muir2009internet, gill2010dude, abdou2017accurate, abdou2015cpv, weinberg2018how}. Participants in some blockchain based services have incentives to obfuscate their true geolocations (e.g., Helium ~\cite{helium}). The location-based incentives lead to adversarial behavior by the participants, including untrusted landmarks and  collusion among network participants, causing significant challenges to both the phases of delay based geolocation. 

\paragraph{Our contributions} We systematically study Byzantine fault tolerant delay based geolocation, including concrete modeling of threat vectors and a mathematical formulation that allows the characterization of security and accuracy of location identification. In this framework, a node desiring to substantiate its claimed location (termed ``Waldo'') subscribes to a group of trustless challengers, possessing self-reported, but unverified locations,  tasked with resolving the geolocation query: ``Where is Waldo?''. The key contribution is in Byzantine resistant inferences from data constrained by (Euclidean) geometry, that we call {\em Byzantine Fortified Trigonometry} (BFT).  We introduce new protocols for both stages of the geolocation process, detailed as follows: 

\begin{enumerate}
    \item[-] \noindent {\bf Proof of Internet geometry (PoIG)}. We propose a Proof of Internet Geometry (PoIG) protocol for converting delay measurements into distance estimates resistant to Byzantine manipulations. During this stage, challengers perform ping delay measurements among themselves to establish a monotone curve that maps any given delay to the largest possible distance.  To mitigate the inflated delay and distance claimed by malicious challengers, we design a robust filtering process that excludes data with the largest distance to delay ratios, achieving 95\% accuracy when the majority are honest (Figure~\ref{fig:acc}). Additionally, we utilize robust matrix completion techniques to lower the costs associated with ping measurements and to clean the data from honest participants.
    \item[-] \noindent {\bf Proof of location (PoLoc)}. In the second stage, we present a  Proof of Location protocol (PoLoc) within a fully Byzantine context, taking into account both Waldo's and the challengers' Byzantine behaviors. The protocol yields the maximum possible deviation of Waldo from its proclaimed location, referred to as ``uncertainty'', alongside with a cryptographic proof representing the entire feasible region (see Figure~\ref{fig:circles}). The protocol ensures that for an honest Waldo, the uncertainty accurately reflects the optimal deviation bound based on the measurements, and for a Byzantine Waldo, the actual location cannot be further from the claimed location than the uncertainty allows.  The proof comprises signed pings and measurements between Waldo and the challengers, facilitating public validation. We analyze that lower uncertainty can be achieved by improving the accuracy of PoIG, limiting the fraction of Byzantine actors, and increasing the number of challengers .
\end{enumerate}

\paragraph{Implementation and evaluations} We implemented both PoIG and PoLoc protocols and integrated them into a fully functional delay-based geolocation service, with global coverage (with specific focus on Asia, North America and Europe) and integrated into major blockchains (Ethereum and Solana). 
We conducted a comprehensive evaluation of our protocol to demonstrate its accuracy and robustness against a variety of practical adversarial behaviors. When examining accuracy, the PoIG protocol outperforms the ``bestline'' fit used in CBG~\cite{gueye2004constraint} by 10\% accuracy. And the PoLoc protocol achieves an uncertainty of less than 100 km for about 45\% of honest Waldos and an uncertainty of less than 1000 km for about 95\% of nodes, using 450 challengers. When examining the robustness of the implementation, our protocol effectively detects measurements manipulation during PoIG phase, and location spoofing and VPN usage during PoLoc phase. 

\paragraph{ Organization} The rest of the paper is organized as follows: In Section~\ref{sec:related_work}, we compare our protocol with prior work, highlighting the unique contributions of our approach. We formally define the system model and security assumptions in Section~\ref{sec:model}. In Section~\ref{sec:PoLoc}, we outline the main PoLoc protocol and explain how trigonometry is employed to ensure robust location estimation. Section~\ref{sec:PoIG} introduces two variants of the PoIG protocols, designed to operate under different trust resources. We present the results of our evaluation in Section~\ref{sec:evaluation}. Finally, we conclude the paper in Section~\ref{sec:conclusion}, summarizing our findings and suggesting directions for future research.

%% file: 02_related_work.tex
\section{Related Work}
\label{sec:related_work}

~\subsection{Delay Based Geolocation}
In delay-based geolocation, ping delay is measured between hosts with known locations, known as landmarks, and the target. These methods vary in their approach to mapping delay to distance. Constraint-Based Geolocation (CBG)~\cite{gueye2004constraint} assigns delay to distance individually for each landmark. It establishes a straight line, termed the ``bestline'' fit, for the pairs of delay-distance points from a given landmark to all other landmarks, maximizing the possible distance any other landmark can be from a given landmark for a specific delay. Octant~\cite{wong2007octant} employs a convex hull around pairs of delay-distance points to estimate the minimum and maximum distance the target could be from a given landmark.

Unlike deriving a delay to distance mapping separately for each landmark, Spotter~\cite{laki2011spotter} develops a single delay-distance mapping for all landmarks through a probabilistic delay-distance model. Geoping~\cite{padamanabban2001determining} implements two distinct delay-based methods for geolocation. Initially, it identifies the location of the target as the location of the landmark with the shortest delay. Subsequently, it employs a probability distribution-based approach for estimating distance based on delay. Studies by~\cite{arif2010internet, arif2010geoweight} also utilize a probability distribution-based mapping from delay to distance, whereas~\cite{dong2012network} adopts a segmented polynomial regression model. An investigation by~\cite{ziviani2005improving} explores how the accuracy of delay-based approaches is influenced by the placement of landmarks.

Delay-based geolocation accuracy can be further enhanced through integration with other information, such as routing information~\cite{jayant2004toward, katz2006towards}, internet topology~\cite{katz2006towards} and region labels ~\cite{hong2023cheap}.

Our protocol leverages the delay-based geolocation approach. However, as detailed in Sec.~\ref{sec:PoLoc}, we ensure that the estimated distance from the challenger to Waldo is always greater than the actual distance. Each challenger, therefore, utilizes a monotone delay to distance mapping that yields the maximum possible distance for a given delay. Our method resembles CBG~\cite{gueye2004constraint}, but through the use of a monotone mapping, we demonstrate an improvement in performance compared to CBG, as discussed in Sec.~\ref{subsubsec:cbg}.

\subsection{Other Approaches}

Besides measurement-based geolocation, other techniques have been employed for geolocation.~\cite{guo2009mining} applies web mining to scrape web pages containing location information of web servers, thus tracking the cities hosting web servers.~\cite{dan2018distributed} proposes a machine learning-based solution for mapping DNS names to locations, while~\cite{dan2016improving} enhances location estimates by mining query logs. GeoTrack~\cite{padamanabban2001determining} maps router labels to locations.~\cite{lee2016ip} utilizes crowdsourcing through an Internet speed measurement tool for geolocation.~\cite{dan2021ip} employs IP interpolation, assuming IPs within the same range are collocated. Additional methods include the use of host names inferred from traceroute, DNS mapping, and others for location identification~\cite{padmanabhan2001investigation, chandrasekaran2015alidade, scheitle2017hloc, wang2011towards}. Some studies apply machine learning to enhance location estimates~\cite{youn2009statistical, eriksson2010learning, eriksson2012posit, jiang2016ip}. There are also various services, both free and paid, that maintain IP to location databases~\cite{poese2011ip, hostip, ipinfodb, ip2location, maxmind}.

While these alternative approaches can complement delay-based geolocation to improve accuracy, they do not offer Byzantine fault tolerance. Therefore, their adaptability in a Byzantine setting remains uncertain.

\subsection{Byzantine Geolocation}

Our focus primarily lies on Byzantine fault tolerance for geolocation. In this context,~\cite{abdou2015cpv, weinberg2018how, kohls2022verloc,maram2021goat} are most closely related to our work.~\cite{abdou2015cpv} provides location guarantees for a Byzantine Waldo using the triangle inequality to offer location guarantees. However, they produce a triangular region where Waldo could likely be located, in contrast to our protocol which determines the intersection of circles centered at the challengers. They assume Waldo cannot inflate delays, a restriction we do not impose.~\cite{weinberg2018how} addresses the geolocation of VPN servers potentially misrepresenting their locations. While their scenario allows for inflated ping delays involving a VPN server, they do not offer location guarantees for a Byzantine Waldo. ~\cite{kohls2022verloc} proposes a decentralized geolocation verification protocol that outputs the most possible location and a confidence score, their information propagation model is built by a group of trusted servers. Recent work~\cite{maram2021goat} explores the geolocation problem in decentralized file storage service with trusted anchors providing timestamps. Hence, our work differs from these studies as our threat model varies, and we provide performance guarantees for our protocol.

%% file: 03_model.tex
\section{Security Model}
\label{sec:model}

\subsection{Problem Setting}
The problem setting that we consider involves a node, Waldo, who claims a particular location, $\tilde{P}$, which includes latitude and longitude coordinates. Waldo's true location is denoted as $P$. We have $\tilde{P} = P$ if Waldo is honest; otherwise, $\tilde{P} \neq P$.

A group of $N$ challengers are geographically distributed following a uniform distribution. Each challenger reports their location, $C_i$ (where $i$ ranges from 1 to $N$). Both Waldo and the challengers are connected to the Internet, allowing for direct communication. The challengers can measure the ping delay to Waldo or other challengers through application layer pings. These pings are cryptographically signed by both parties using their private keys, as part of a public key infrastructure (PKI), serving as proof of communication. The public keys are broadcast before the protocol starts.

The core issue we address is quantifying the extent of deviation that $\tilde{P}$ can have from Waldo's real location, $P$. This deviation, denoted as uncertainty $R = \|P-\tilde{P}\|$, is expressed in kilometers (km) and signifies the furthest possible distance that $\tilde{P}$ might diverge from $P$. We consider the maximum $R$ in any direction, termed as $R^{*}$. The formal definition of $R^*$ is illustrated in Section~\ref{subsec:proof-of-location}.

\subsection{Threat model}
\label{subsec:threat-model}
We adopt a Byzantine failure model where the adversary can corrupt a certain fraction of participants, causing them to deviate from the protocol arbitrarily. However, the security of PKI prevents them from forging signatures. A corrupted Byzantine Waldo may claim a false location ($\tilde{P} \neq P$) and delay or refuse to send a response upon receiving an application layer ping from challengers. The adversary can corrupt less than $ \beta$ fraction of challengers uniformly randomly in locations. Byzantine challengers may also claim false locations and manipulate the measurement results. Note that the ping delay measured by an honest challenger can only be inflated if Waldo or the target challenger are Byzantine, as they cannot send back a response before receiving the actual packets due to the unforgeability of signatures. However, the reported delay can be both inflated or deflated when the source challengers are Byzantine.



%% file: 04_protocol.tex
\section{Proof of Location Protocol }
\label{sec:PoLoc}
In this section, we detail the methodology our protocol employs to quantify the uncertainty in Waldo's geolocation and analyze various factors influencing this uncertainty. 

\subsection{Protocol Overview}
\label{subsec:proof-of-location}
When Waldo wishes to prove its location, $\tilde{P}$, a PoLoc challenge is initiated. All challengers within a predetermined distance, $X$ kms from $\tilde{P}$, are selected to validate it. The determination of $X$ is discussed in Section~\ref{sec:PoIG}.

Suppose $n$ challengers are selected to verify Waldo. Each challenger initiates communication with Waldo by sending a ping packet and waiting for its response. They measure the round-trip time (RTT), $t_{i}$ for the $i$-th selected challengers, from when the ping packet was sent to when the response was received. A large number of ping packets are exchanged, and the minimum RTT from all measurements is selected. Using a large number of packets helps account for any intermittent network congestion. As previously stated, Waldo can delay the ping response, meaning the delay, $t_{i}$, can be inflated. If challenger $i$ is also corrupted, they can collude to report an arbitrary $t_{i}$.

Each challenger $i$ has a delay-to-distance mapping, denoted as $F_i$, which converts measured RTT into physical distance. $F_i$ is derived from a proof of internet geometry (PoIG) phase, discussed in Section~\ref{sec:PoIG}. We explore two variants of PoIG protocols: one requiring trusted challengers and the other applying robust filtering techniques to minimize the trust assumption. To increase the efficiency and reduce noises, we employ robust matrix completion to fill the RTT measurements. With their delay-to-distance mapping, challenger $i$ can estimate how far Waldo is from itself, i.e., $d_i = F_i(t_i)$. As demonstrated later in Section~\ref{sec:trigonometry}, the challenger can then estimate the distance between $\tilde{P}$ and $P$ in every direction to determine the boundary of the output region.

\subsection{The Trigonometry in Proof of Location}
\label{sec:trigonometry}

\begin{figure}[h]
    \centering
    \begin{minipage}[b]{1.0\linewidth}
    \centering
    \includegraphics[width=1.0\textwidth]{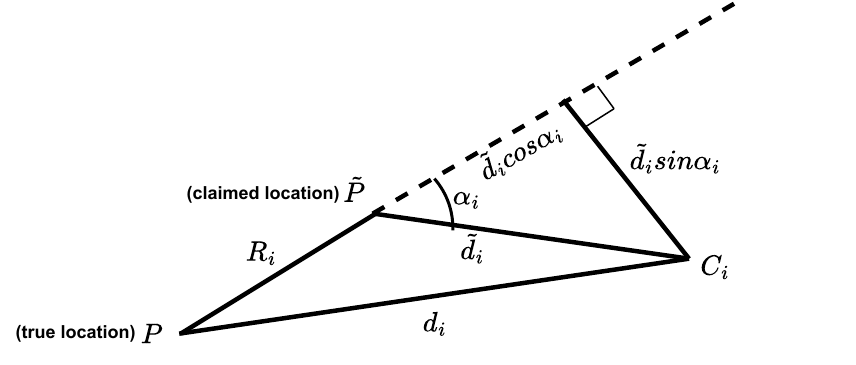}
    \begin{tabular}{|c|c|}
        \hline
        $P$ & true location of Waldo  \\
        $\Tilde{P}$ & claimed location of Waldo  \\
        $C_{i}$ & challenger $i$  \\
        $R_{i}$ & uncertainty in geolocation of Waldo by $C_{i}$  \\
        $d_{i}$ & distance between $C_{i}$ and $P$  \\
        $\Tilde{d_{i}}$ & distance between $C_{i}$ and $\Tilde{P}$  \\
        $\hat{d_{i}}$ & distance estimated between $C_{i}$ and $P$ from latency  \\
        \hline
    \end{tabular}
    \end{minipage}
    \caption{Bound on uncertainty.}
    \label{fig:bound}
\end{figure}

\input{PoLoc-alg}

Figure~\ref{fig:bound} illustrates the situation where an honest challenger $i$ is measuring $P$, the true location of Waldo, while $\tilde{P}$ is the claimed location. The distance between $P$ and $\tilde{P}$, represented as $R_{i}$, is the uncertainty in a certain direction with angle $\alpha_i$ relative to the direction $C_i\tilde{P}$.

The PoLoc protocol is detailed in Algorithm~\ref{alg:PoLoc}, containing a ping phase conducted by each challenger individually, and a proving phase when measurements are collected to evaluate the location. In the ping phase, the challenger $i$ measures the minimum RTT $t_{i}$ between its own location $C_i$ and the actual location of Waldo $P$. It then estimates the distance by calculating $\hat{d}_i = F_i(t_i)$. Here, we assume the challenger knows the delay-to-distance mapping $F_i$. The process of calibrating the mapping through PoIG is discussed in Section~\ref{sec:PoIG}.

We denote $\hat{d_{i}}$ as the distance between $C_i$ and the claimed location of Waldo. As mentioned, a Byzantine prover can only inflate the RTT when the challenger is honest (by potentially delaying the response), which means
\begin{equation}
\hat{d_{i}} \ge d_{i}. \label{eq:inflation}
\end{equation}

From Figure~\ref{fig:bound}, we have:

\begin{equation}
d_{i}^2 = (R_{i} + \tilde{d_{i}}\cos\alpha_{i})^2 + \tilde{d_{i}}^2\sin^2\alpha_{i} \label{eq:tri}
\end{equation}

where $\alpha_{i}$ is the angle between the lines joining challenger $C_{i}$ and $\tilde{P}$, and $P$ and $\tilde{P}$. Substituting Equation~\eqref{eq:tri} in Equation~\eqref{eq:inflation} and simplifying, we get:

\begin{equation}
\label{eqn:bound}
R_{i} \le \sqrt{\hat{d_{i}}^2 - \tilde{d_{i}}^2\sin^2\alpha_{i}} - \tilde{d_{i}}\cos\alpha_{i}.
\end{equation}

Thus, Equation~\ref{eqn:bound} provides the bound on the uncertainty of Waldo in terms of the estimated distance between the challenger and Waldo, $\hat{d_{i}}$.

This uncertainty calculation is for a single challenger for angle $\alpha_{i}$. To find the uncertainty in every direction with $\alpha_{i}$ varying from 0 to $2\pi$, we consider that $R_{i}$ will vary for different challengers at different angles $\alpha_{i}$. We aim to derive the maximum uncertainty across all directions, denoted as $R^{*}$.

To obtain $R^{*}$, we first calculate $R_{i\theta}$ for every challenger, which is the maximum $R_{i}$ for challenger $i$ in direction $\theta$ (with respect to geographic north). We denote $\gamma_{i}$ as the angle between $C_{i}$ and $\tilde{P}$ with respect to geographic north. Then, $R^{*}$ can be easily calculated by substituting $\alpha_i$ with $\gamma_i - \theta$ in Equation~\ref{eqn:bound}.

The goal is to first find the minimum uncertainty in each direction reported by honest challengers. It is important to note that if a challenger is also corrupted, the assumption in Equation~\eqref{eq:inflation} no longer applies. In other words, a Byzantine challenger $c_j$ can report an arbitrarily small $\hat{d_j} < d_j$. To prevent malicious measurements from influencing the final uncertainty, in each direction, we exclude the $\beta n$ smallest uncertainty values calculated from different challengers. For simplicity, we define $S_k$ as the $k$-th smallest element in set $S$, and $k$ is the position of the element when the set is ordered. Then, we get Equation~\eqref{eq:smallest}.

We vary $\theta$ from 0 to $2\pi$, which gives us the boundary of the region containing all possible locations where Waldo might be. We then take the maximum across all $R_{\theta}^{*}$ as the final uncertainty, $R^{}$. The computation of the final $R^{*}$ is fomulated in Equation~\eqref{eq:output}.

\paragraph{Relation to circle intersection}
The trigonometry calculated above can be visualized to better derive the potential true location of Waldo. Essentially, Equation~\ref{eq:inflation} represents a circle with radius $\hat{d}_i$ centered at the challenger $i$'s location. To understand what the outputs of our protocol capture, let's start by considering examples with $\beta = 0$, meaning all challengers are honest. In this case, Equation~\ref{eq:smallest} outputs the minimum uncertainty in all directions as the boundary of the region. It is intuitive to see that the potential region where Waldo might be located is equivalent to the intersection of the circles centered at each challenger with radius $\hat{d}_i$. Figure~\ref{fig:intersection} shows an example of PoLoc outputs and circle intersections when $n=3$ and $\beta=0$.

\begin{figure}[h]
    \centering
    \includegraphics[width=0.45\textwidth]{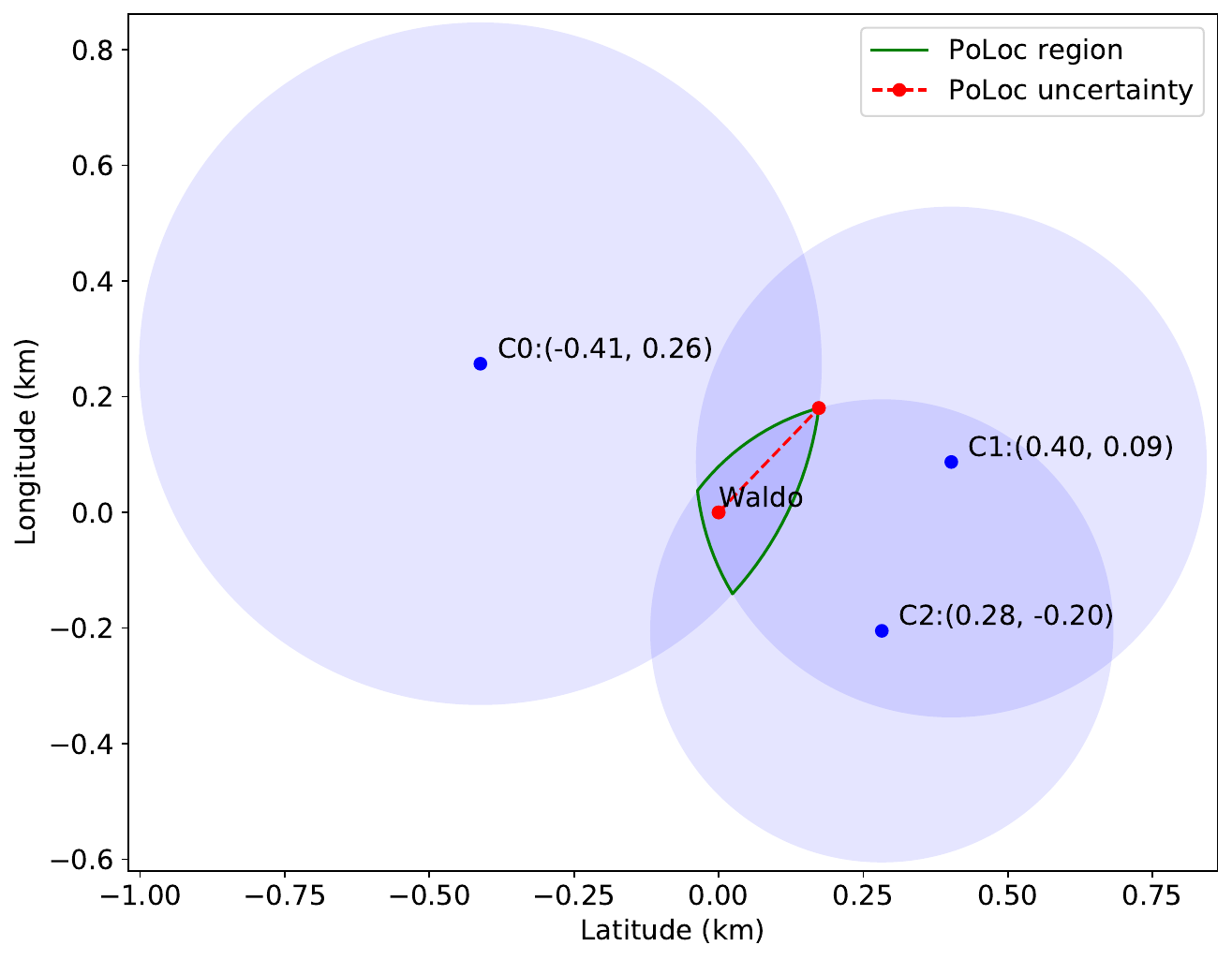}
    \caption{PoLoc outputs and the intersection of the circles centred at challengers ($n=3, \beta=0$).}
    \label{fig:intersection}
\end{figure}

Now, considering the case when $\beta > 0$, since the first $\beta n$ smallest uncertainty values are ignored, the output region becomes larger than in the case where $\beta = 0$. This exclusion can be considered as that we can output all regions that are output by at least $(1-\beta)n$ challengers, in other words, the area covered by at least $(1-\beta)n$ circles. Figure~\eqref{fig:intersection-robust} shows an example of our PoLoc outputs and circle intersections when $n=10$ and $\beta=0.3$. The malicious Waldo tries to claim a different location with the help of three corrupted challengers (whose region are marked in yellow).

\begin{figure}[h]
    \centering
    \includegraphics[width=0.45\textwidth]{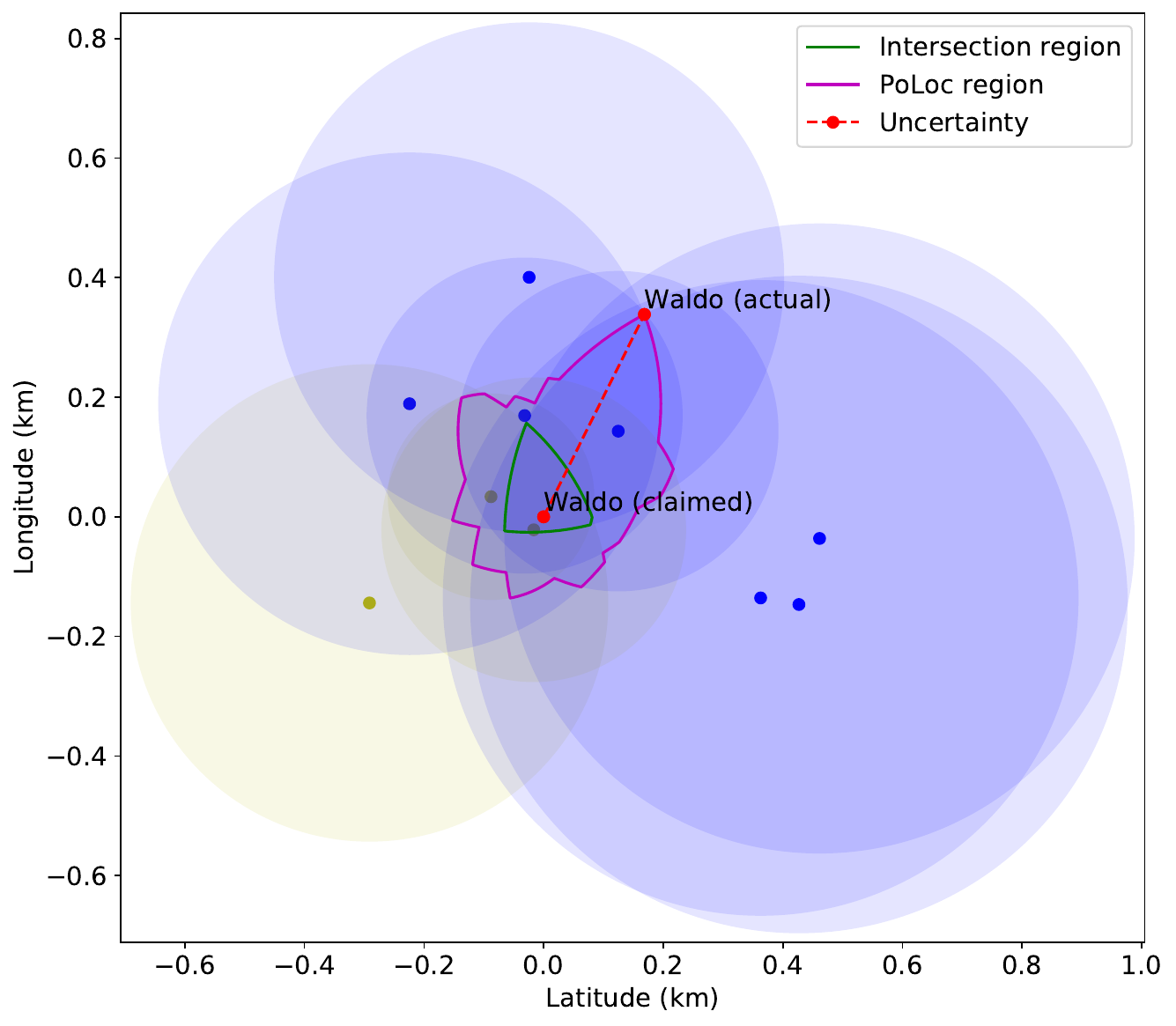}
    \caption{PoLoc outputs and the intersection of the circles centred at challengers ($n=10, \beta=0.3$).}
    \label{fig:intersection-robust}
\end{figure}

\subsection{Protocol Guarantees}

 Byzantine Waldo may have incentives to claim a very different location (for anonymity, for earning more rewards), so the eventual uncertainty $R^*$ output by PoLoc protocol captures the maximum deviation that Waldo's true location can have from the claimed location. We first prove the soundness of the protocol.

\begin{theorem}
    (Soundness) The actual location of Waldo will not be farther away from its claimed location than the uncertainty $R^*$ calculated in Equation~\ref{eq:output}, even if Waldo is Byzantine.
\end{theorem}

\begin{proof}
    We denote $\theta'$ as the direction in which Waldo is located with respect to the claimed location such that $|P - \tilde{P}| > R^* > R_{\theta'}^{*}$. According to Equation~\ref{eq:smallest}, $R_{\theta'}^{*}$ is the $\beta n$ smallest uncertainty measured by challengers, so Equation~\ref{eq:inflation} doesn't hold for $\ge \beta n$ challengers. However, the adversary can corrupt less than $\beta n$ challengers,  so at least one honest challenger deflates the measurement, which contradicts the assumption. As a result, there does not exist a direction in which the uncertainty is larger than $R^*$.
\end{proof}

It is obvious that a larger uncertainty indicates a lower accuracy of the protocol. We aim to prove that our PoLoc protocol provides the optimal uncertainty given the measurement data.

\begin{theorem}
    (Completeness) The uncertainty $R^*$ calculated in Equation~\ref{eq:output} is optimal given $R_{i\theta}, F_i$ for $i\in[1,n], \theta\in[0,2\pi]$.
\end{theorem}

\begin{proof}
    We prove that uncertainty $R^*$ is optimal by contradiction. Assume that there exists some uncertainty $R' < R^*$ that better captures the largest deviation of Waldo, i.e., in all direction $\theta\in [0,2\pi]$, the distance between the actual location and the claimed location of Waldo $\|P - \tilde{P}\|\le R' < R^*$. However, according to Equation~\ref{eq:output}, let's denote $\theta^* = \arg\max_\theta \{R_{i\theta} | i\in [1,n]\}_{\beta n}$. At direction $\theta^*$, we calculate the uncertainty of Waldo in terms of all challengers $i$ as $R_{i\theta^*}$, there exists a situation where all Byzantine challengers output 0, and some honest challenger $i^*$ who outputs the smallest $R_{i^*\theta^*} = R^*$ accurately measures $P$ ($\hat{d}_{i^*} = d_{i^*}$ in Equation~\ref{eq:inflation}). Following Equation~\ref{eq:tri} and \ref{eqn:bound}, we have 
    \begin{equation}
        |P - \tilde{P}\| = R^* > R'
    \end{equation}
    which contradicts to the assumption that $\|P - \tilde{P}\|\le R'$ holds in any cases.
\end{proof}

Though the closed-form of $R^*$ doesn't exist, we have seen that $R^*$ resides in one of the arcs which make up the boundary of the region intersected by exactly $(1-\beta)n$ circles. To better understand how uncertainty changes with respect to other system parameters such as the adversarial fraction $\beta$, the number of challengers $n$ and the accuracy of delay-to-distance mapping $F_i$, for simplicity we simulate the case where the mapping outputs the ground truth distance with an one-sided Gaussian noise,
\begin{equation}
\label{eq:simulation}
    F_i(t_i) = d_i + |\mathcal{N}(0,\sigma^2)|
\end{equation}
where $\sigma^2$ is the variance of the Gaussian noise.

\begin{figure}[h]
    \centering
    \includegraphics[width=0.45\textwidth]{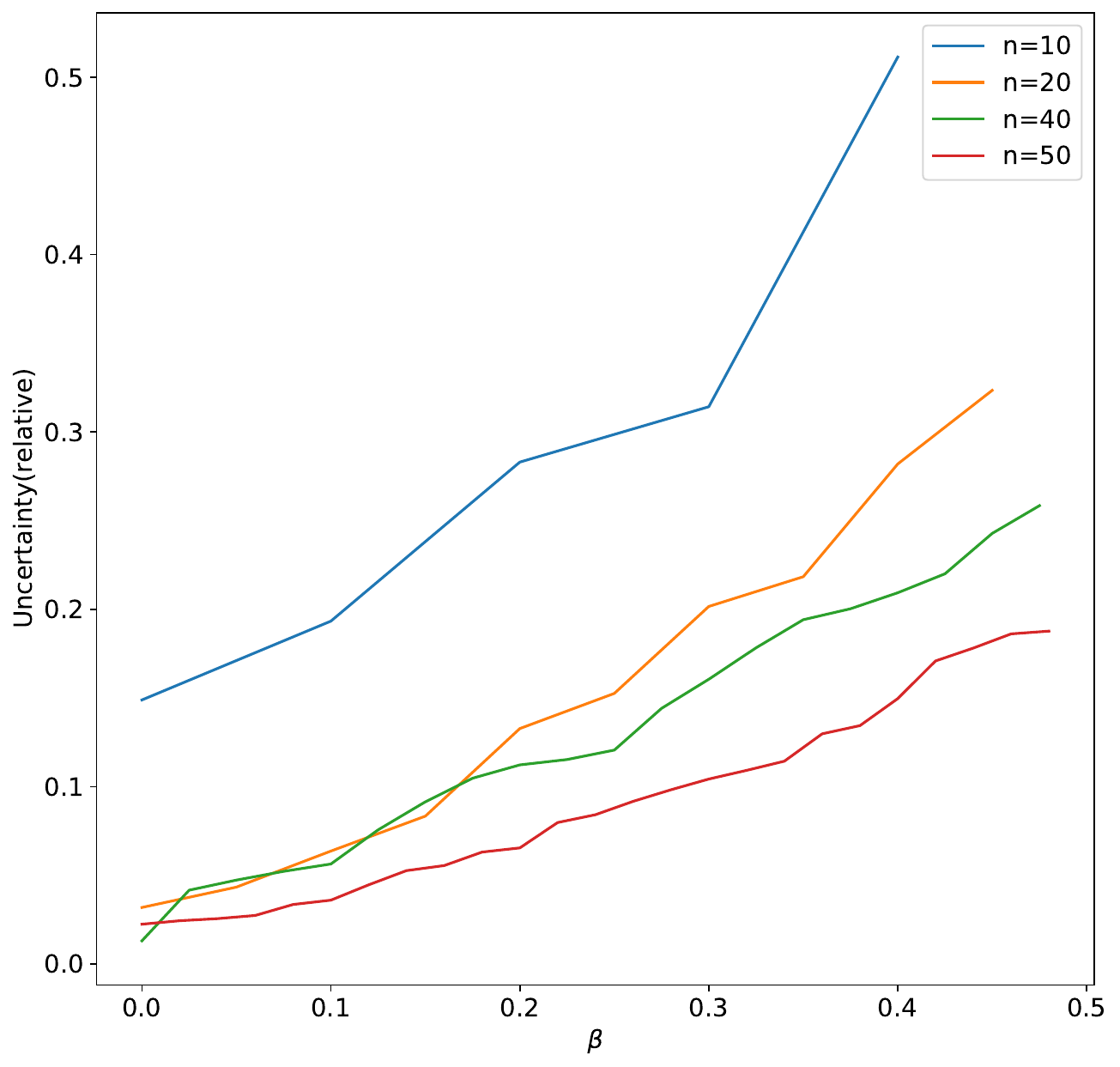}
    \caption{The effect of number of challengers ($n$) to uncertainty and Byzantine fraction ($\beta$) curve. 
}
    \label{fig:k}
\end{figure}

We set Waldo's actual location at the origin and randomly sample $n$ challengers within a unit circle around this point. Figure~\ref{fig:k} displays the simulation outcomes for $n=10, 20, 40, 50$ challengers, with the fraction of Byzantine challengers $\beta$ varying from $0.1$ to $0.5$. For each parameter set, we conduct $50$ experiments to determine the average uncertainty $R^*$. The results indicate that, generally, the uncertainty decreases with an increasing number of challengers. A higher $\beta$ value increases the uncertainty across all scenarios. Figure~\ref{fig:var} examines the impact of different variances in delay-to-distance mappings on the uncertainty, showing that less noisy mappings lead to smaller uncertainty values, thereby enhancing the PoLoc protocol's performance.

\begin{figure}[h]
    \centering
    \includegraphics[width=0.45\textwidth]{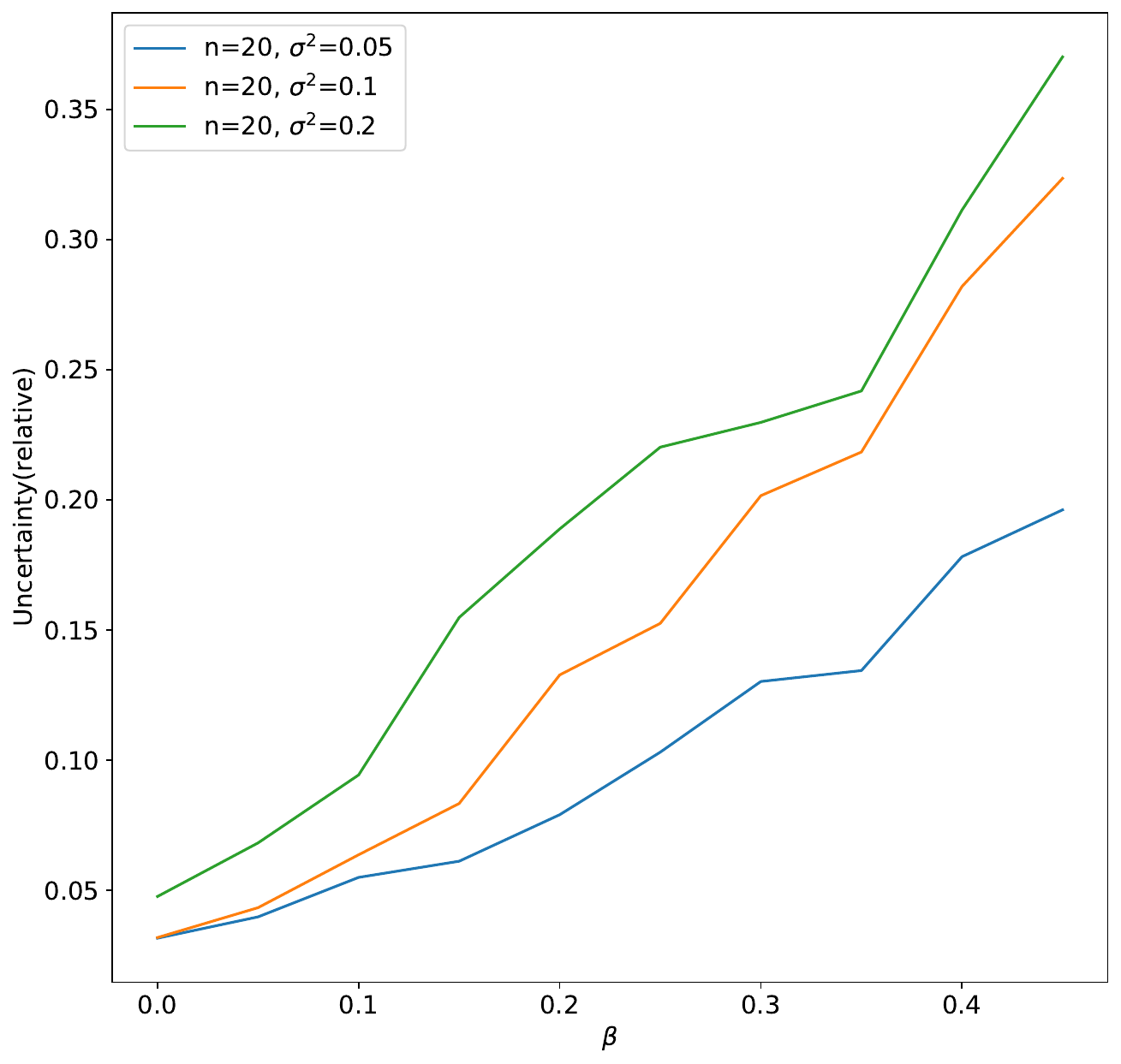}
    \caption{The effect of variances of delay-to-distance mapping to uncertainty and Byzantine fraction ($\beta$) curve. 
}
    \label{fig:var}
\end{figure}

%% file: PoLoc-alg.tex
\begin{algorithm}[h]
\protect\caption{Proof of Location Protocol}
\label{alg:PoLoc}

\textbf{Input}:  $n$, $\beta$, $[F_1, F_2, \cdots, F_n]$, $\tau$, $[C_1, C_2, \cdots, C_n]$, $\tilde{P}$.

\textbf{Initialize}: $\tilde{d}_i = \|C_i - \tilde{P}\|$

\textbf{Ping phase for challenger $i$}: For $i=1,\ldots,n$, challenger $i$ sends a signed ping packet to Waldo, who also signs the packet as response. On receiving the response, challenger $i$ measures the time elapses and repeat $q$ times to get a minimum $t_{i}$. Then it calculates $\hat{d}_i = F_i(t_i)$.

\textbf{Proving phase}: On receiving all distance estimations or at least $(1-\beta)n$ measurements and a timer $\tau$ expires (note that an unobserved $\hat{d}_i = 0$), we calculate the following value for each challenger $i\in [1,n]$ at each direction $\theta\in [0, 2\pi)$: 
\begin{equation}
R_{i\theta} = \sqrt{\hat{d_{i}}^2 - \tilde{d_{i}}^2\sin^2(\gamma_{i}-\theta}) - \tilde{d_{i}}\cos(\gamma_{i}-\theta)
\end{equation}
where $\gamma_{i}$ is the angle between $C_{i}$ and $\tilde{P}$ with respect to geographic north. 
Then we exclude the $\beta n$ smallest values in each direction $\theta$, let $S_k$ represent the $k$-th smallest element in set $S$, we compute:
\begin{equation}\label{eq:smallest}
    R_{\theta}^{*} = \{R_{i\theta} | i\in [1,n]\}_{\beta n}
\end{equation}

\textbf{Output}: 
\begin{equation}\label{eq:output}
R^{*} = \max_{\theta \in [0, 2\pi]} R_{\theta}^{*}
\end{equation} 
\end{algorithm}

%% file: 05_calibration.tex
\section{The Geometry of Internet Protocol}
\label{sec:PoIG}

\global\long\def\P{\mathcal{P}}
\global\long\def\I{I}

As outlined in the previous section, an essential requirement of our protocol is that the estimated distance of the honest challenger to Waldo should be greater than the actual distance (equation~\ref{eq:inflation}). Since the estimated distance is calculated from delay-to-distance mapping, it is crucial for a challenger to obtain a mapping that converts delay to the largest likely possible distance that the challenger can be from any Waldo. The mapping essentially captures the underlying geometry of the internet.

Therefore, our PoLoc protocol includes a one-time calibration phase to learn the geometry of the internet, referred to as Proof of the Internet Geometry (PoIG) protocol. It is conducted by each challenger before participating in PoLoc protocol for proving any Waldo's location. We consider two variants of the PoIG protocol based on the availability of a pool of trusted servers.

\subsection{Partially Trusted PoIG Protocol} 
In the first variant, we assume the existence of a pool of trusted network servers that respond to ping packets from challengers. Challengers measure the ping delay to these servers, whose locations are known. The challengers then calculate the distance to these servers based on their reported locations and map ping delays to distances. While Byzantine challengers can still inflate delays, each challenger collects their own data to learn the mapping, thereby not being affected by Byzantine attacks from other challengers.

To ensure that for a given ping delay, the challenger outputs the maximum possible distance that any Waldo can be from it,  we employ a a monotone mapping approach which is distinct from those used in prior research, such as the linear or convex hull mappings~\cite{gueye2004constraint, wong2007octant}. 

\begin{figure}[h]
    \centering
    \includegraphics[width=0.45\textwidth]{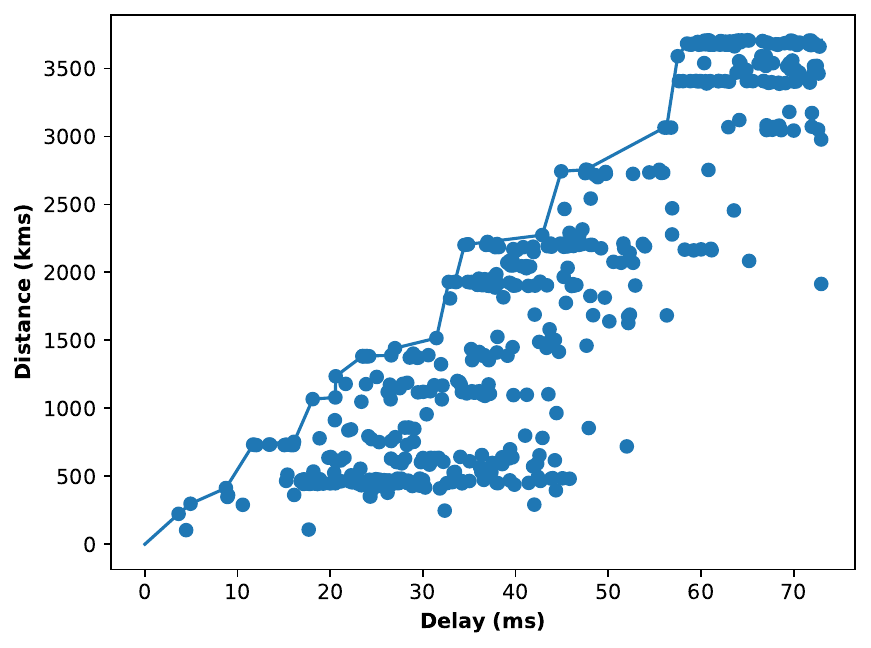}
    \caption{Monotone mapping from delay to distance.}
    \label{fig:mapping}
\end{figure}
A challenger plots all measured delay-distance pairs on the X-Y coordinate plane, with round-trip delay on the X-axis and distance on the Y-axis, as shown in Figure~\ref{fig:mapping}. To construct the monotone mapping, for a given delay, $t_1$, we select the maximum distance observed so far, $d_{1\max}$. For any subsequent time instance, $t_2$, where $t_2 > t_1$, we connect $d_{1\max}$ to $d_{2\max}$, only if $d_{2\max} \ge d_{1\max}$. If $d_{2\max} < d_{1\max}$, then we seek a subsequent time instance $t_3$, or later, which meets this monotone requirement. 
Figure~\ref{fig:mapping} illustrates this monotone mapping with a solid line, connecting points on the outer periphery to result in a monotone map from delay to distance such that the distance is a non-decreasing function of increasing delay. This approach provides a tighter fit than linear mappings used by CBG~\cite{gueye2004constraint} as we illustrate in Sec.~\ref{subsubsec:cbg}.

\paragraph{Server selection} We observe that the accuracy of delay measurements increases with the distance between nodes, so we establish an upper limit of $X$ kms when selecting servers for constructing the mapping. Our data indicate that an X value of 2000 kms roduces favorable results. Figure~\ref{fig:mapping} illustrates the significant variation in distance for a given delay; for instance, a 30 ms delay correlates to distances ranging from approximately 400 kilometers to 1500 kilometers. Such disparity is attributable to various factors, including the layout of fiber optic links and the peering relationships between ISPs~\cite{kasiviswanathan2011geography, eriksson2013understanding}.  

\begin{figure}[h]
    \centering
    \includegraphics[width=0.45\textwidth]{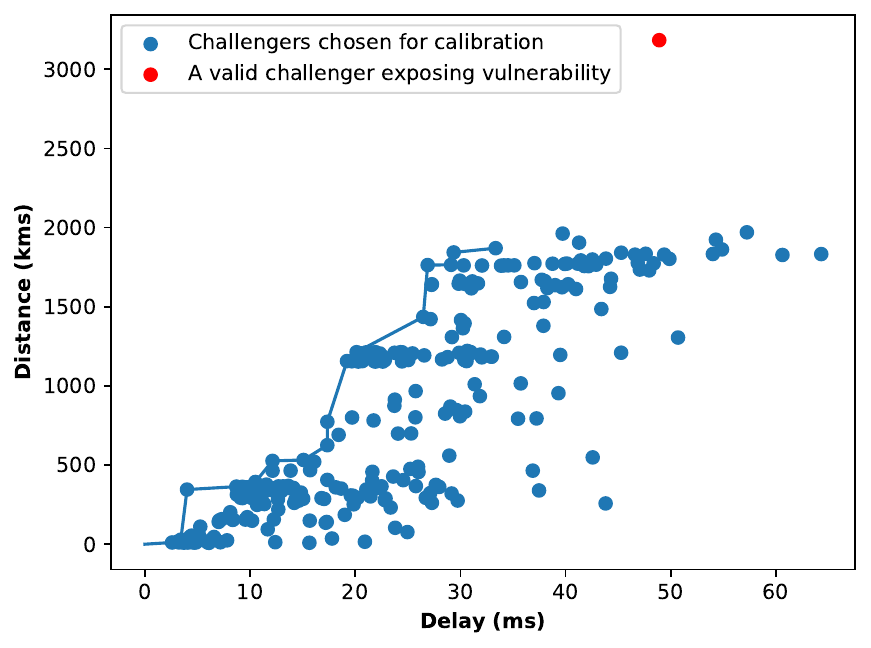}
    \caption{A possible attack.}
    \label{fig:attack}
\end{figure}

This substantial variability introduces a potential attack: challengers located within $2000$ kms of Waldo's claimed location might, in reality, be further than $2000$ kms from Waldo's actual location. An example is highlighted in Figure~\ref{fig:attack}. If we only consider all servers within $2000$ kms of the challenger, the greatest delay within this radius is approximately 65 ms. Servers with delays less than $65$ ms and distances larger than $2000$ kms (like the one depicted in red) are excluded. We account for this attack by adding these additional measurement points. The updated protocol works as follows: Firstly, we compute the maximum delay measured by servers with at most $X$ kms distance, denoted as $t_{max}$, i.e., about 65 ms in Figure~\ref{fig:attack}. Then we pick all other servers that experience a ping delay less than $t_{max}$ regardless of their distances to the challenger. Consequently, the previously excluded red server, depicted in Figure\ref{fig:attack}, is now selected as shown in Figure\ref{fig:mapping}. Indeed, servers located as far as 3500 kms away have been incorporated into the process in Figure\ref{fig:mapping}.

\begin{figure*}
\centering
\begin{subfigure}{.45\textwidth}
  \centering
  \includegraphics[width=\linewidth]{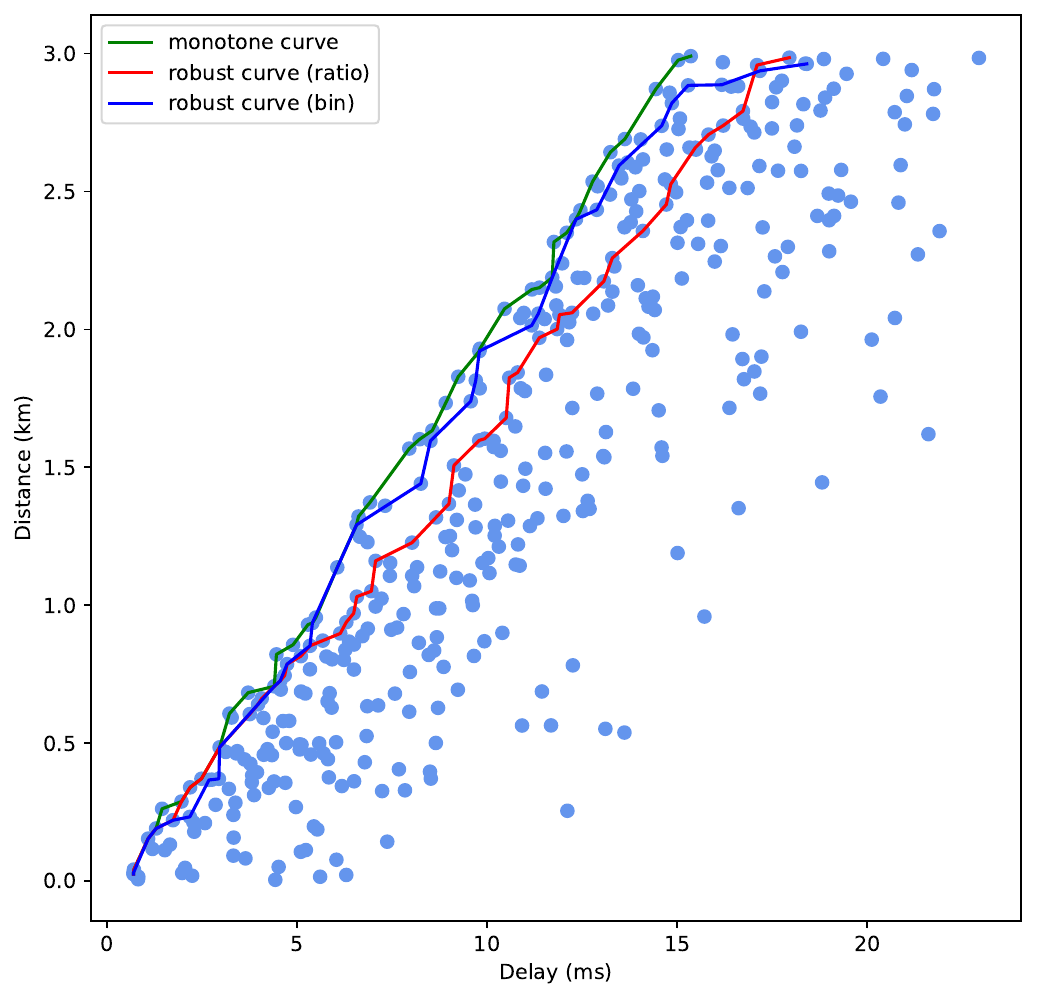}
  \caption{$\beta=0$ }
  \label{fig:calibration-honest}
\end{subfigure}%
\begin{subfigure}{.45\textwidth}
  \centering
  \includegraphics[width=\linewidth]{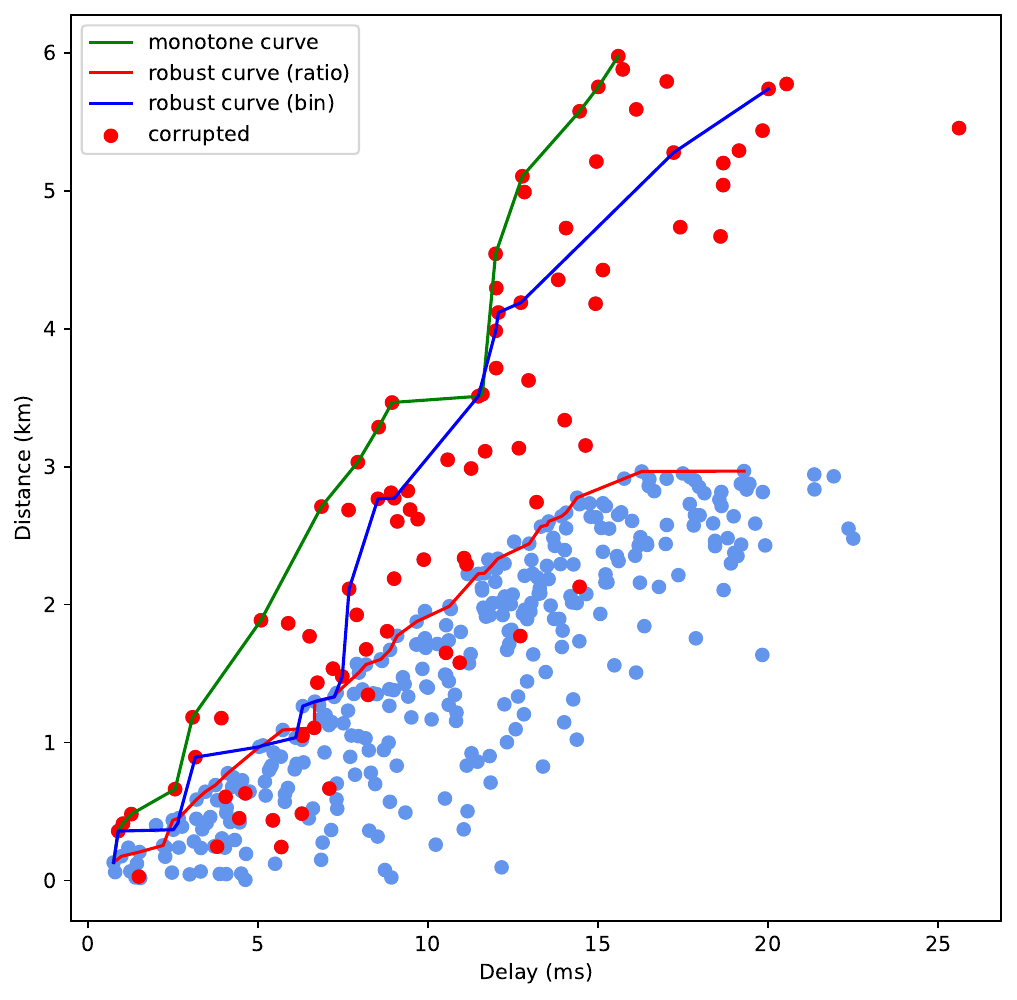}
  \caption{$\beta=0.2$ }
  \label{fig:calibration-byz}
\end{subfigure}
\caption{Comparisons of simulated linear delay-to-distance curves derived by different variants of calibration protocols. }
\label{fig:robust-calibration}
\end{figure*}
\subsection{Trustless PoIG Protocol} 
To adapt to a decentralized environment and further reduce dependency on trust, we explore a trustless version of our protocol. This version operates without relying on external trusted servers, instead utilizing the challengers themselves to construct the mapping. In this setup, each challenger obtains RTT measurements by pinging other challengers. This adjustment introduces a significant challenge: since challengers may falsify their locations and manipulate RTT measurements. For example, Byzantine challengers can report incorrect locations that are farther than their actual positions. This falsely suggests a larger distance to the pinging challenger, thereby skewing the data point to appear above the actual monotone curve.

A preliminary solution to fix the problem is to apply exclusion rules similar to those used in the proving phase of the protocol. Specifically, for each latency (or a narrow latency range), we disregard the the $\beta$ fraction of measurements with the longest distances due to the uniform location distribution of challengers. However, adversaries could strategically distribute Byzantine measurements across a certain distance range. Considering the monotony of the curve, an adversarial data point with large distance and small latency may nullify all the following honest data points with larger latency. Therefore, we propose calculating the distance-to-delay ratio for each measurement and eliminating the $\beta$ fraction of data with the highest ratios.

In our simulation, we generate measurement data based on a linear mapping with added one-sided normal random noise (Equation~\ref{eq:simulation}). Figure~\ref{fig:robust-calibration}  depicts the adjusted mappings after implementing robust exclusion rules.  With $k=400$ challengers participating, we compare scenarios between entirely honest challengers (Figure~\ref{fig:calibration-honest}) and a situation where 20\% ($\beta=0.2$) of the challengers are Byzantine (Figure~\ref{fig:calibration-byz}).
We assume Byzantine challengers double their distance while maintaining the same delay. 

In these illustrations, the basic monotone curve constructed by partially trusted PoIG protocol is shown in green. 
We denote the robust curve, which excludes the $\beta n$ challengers with the highest distance-to-delay ratio in red. And the blue curve represents the measurements after removing a $\beta$ fraction of challengers for each data set within a small delay range (i.e., a bin containing $1/\beta$ nodes). 

The findings suggest that filtering based on the distance-to-delay ratio is exceptionally effective in countering distance inflation attacks, causing only a minimal and acceptable deviation from the baseline curve. This outcome is expected as the data with the most significant ratios tend to lift the curve upward, thus potentially increasing the final uncertainty in the proving phase. However, as seen in Figure~\ref{fig:calibration-honest}, when challengers accurately report their positions, this filtering approach may inadvertently exclude honest data points, resulting in a curve that is less conservative compared to the baseline. This means that for a given latency, the mapped distance might not represent the maximum possible distance.

To address this discrepancy, we propose incorporating a correction factor $\eta$, aimed at preserving the proving phase's soundness with only a minor increase in uncertainty. Figure~\ref{fig:calibration-factor} demonstrates the effect of applying this correction factor to the curve.

\begin{figure}[h]
    \centering
    \includegraphics[width=0.45\textwidth]{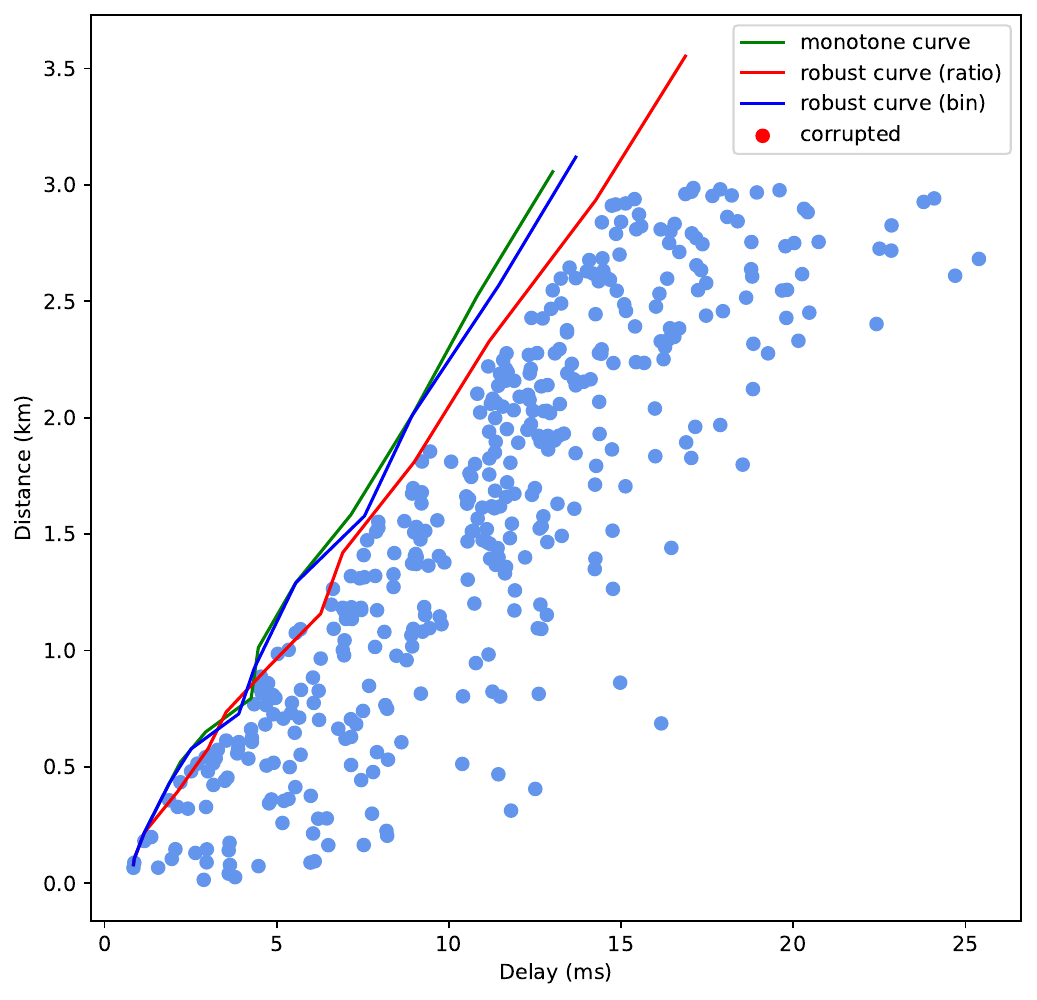}
    \caption{Robust monotone mapping with ratio-based filtering and correction factor $1.2$.}
    \label{fig:calibration-factor}
\end{figure}


\paragraph{Accuracy from regression models}
To understand how these robust filtering methods influence the accuracy of mapping, we start with assuming that the actual distance to delay function is a linear function with normal observation noises, that is 
\begin{equation}
    t_i = cd_i + \epsilon + b
\end{equation}
where $b,c$ are the constant parameters, $d_i$ is the distance to challenger $i$, $\epsilon$ is a Gaussian noise with variance $\sigma^2$. This is a simple linear regression problem and we are looking for an estimate $\hat{c}$ that minimizes the sum of squared residuals:
\begin{equation}
    \hat{c} = \frac{\sum_{i=1}^{k}(d_i - \bar{d})(t_i - \bar{t})}{\sum_{i=1}^{k}(d_i - \bar{d})^2}
\end{equation}
where $\bar{d}$ and $\bar{t}$ are the sample means of $d_i$ and $t_i$, respectively.

We measure the accuracy of the estimator $\hat{c}$ through  the variance as below:
\begin{equation}
    \text{Acc}(\hat{c}) = \frac{1}{\text{Var}(\hat{c})} = \frac{\sum_{i=1}^{k}(d_i - \bar{d})^2}{\sigma^2}
\end{equation}

Considering the uniform distribution of these $k$ selected challengers, the nominator can be written as the product of $k$ and the constant variance of distance distribution, so we have
\begin{equation}\label{eq:acc}
    \text{Acc}(\hat{c}) =  \frac{k\text{Var}(d_i)}{\sigma^2}
\end{equation}
Equation~\ref{eq:acc} shows that with more data samples (larger $k$) and less noises (smaller $\sigma^2$), we get better accuracy of the estimated mapping.

Now consider the constraints that $\beta k$ data samples are Byzantine and removed by some robust filtering method, since Byzantine challengers are also distributed uniformly randomly in locations, it doesn't change the variance of distance to the given challenger, only the number of data samples decreases to $(1-\beta)k$. When Byzantine challengers try to circumvent the detection, honest challengers with small delay are more likely to be removed. As a result, the accuracy with robust protocol is
\begin{equation}
    \text{Acc}_\beta(\hat{c}) \simeq  \frac{(1-\beta)k\text{Var}(d_i)}{\sigma^2}
\end{equation}

The above model assumes the Gaussian noise (two-sided) which may not be the best fit of observation errors for RTT. Now we assume the error $\epsilon$ follows an exponential distribution with mean $1/\lambda$. Since exponential distribution is always positive, our estimated delay with additive noise shows the conservative curve for delay-to-distance mapping, which is consistent with our implemented protocol. Then we solve the maximum likelihood estimation (MLE) by setting the derivative of the log-likelihood function with respect to $\lambda$ and we get for the conservative estimator $\hat{c}$ as follows
\begin{equation*}
    \hat{c} = \min\left(\frac{t_i}{d_i}\right)
\end{equation*}

The distribution of $t_i/d_i$ combines the exponential noise term and the uniformly sampled distance term as denominator, hence we can not provide the direct accuracy. However, given the form of the estimator, the adversarial challengers can easily violate the accuracy by reporting smaller delay/distance ratio. In this case, our ratio-based filtering method effectively filter out these outliers.

\subsection{Delay Matrix Completion}
So far, both partially trusted and trustless PoIG protocols require each challenger to ping other servers or challengers independently, which yields a total of $O(Nk)$ measurements, where $N$ is the total number of challengers and $k$ is the number of selected nearby servers or challengers within $X$ kilometers. To improve the efficiency of calibration protocols, challengers can utilize the measurements conducted by other challengers to construct their own curves through matrix completion methods, the problem is formally described as follows.



Consider we have $m$ challengers, all challengers are within $X$ kms of others and we denote their $m\times m$ squared delay matrix as $M$, where
\begin{equation}
    M_{ij} = t_{ij}^2
\end{equation}
for $(i,j)\in E$, $t_{ij}$ is the RTT measured by challenger $i$ to $j$, $E$ is the index set of measurements. 
We assume that $M$ is a noisy observation of some linear transformation from the squared distance matrix among challengers whose rank is at most $4$~\cite{montanari2010positioning}, hence it has a low-rank structure, we denote the rank as $r$. Due to the existence of Byzantine challengers, a fraction $\beta$ columns (and its corresponding rows) are arbitrarily corrupted. The goal of matrix completion is to infer the delay measurements of non-corrupted challengers and the identities of the corrupted challengers.

Existing works have explored the robust matrix completion problem with noises and corrupted entries~\cite{chen2013low}, with a specific class of work focusing on column-wise corruption~\cite{chen2011robust,chen2015matrix}. In our model an adversarial challenger $i$ participates in all the active measurements it issues and receives, it can corrupt the entire row and column $M_{i*}$ and $M_{*i}$. We follow the paradigm in ~\cite{chen2011robust,chen2015matrix} and decompose $M$ into two parts
\begin{equation}
    M = L + C
\end{equation}
where $L$ is the  matrix with non-corrupted columns whose rank is $r$. At most $(1-\beta)m$ of the columns of $L$ are non-zero. $C$ is the matrix for corrupted columns and at most $\beta n$ columns are non-zero. The observed entries of $M$ are sampled by a Bernoulli model with uniform probability $p$. Note that for a given $C$, we will remove all the rows with index of non-zero columns in $C$ in $L, C, M$ to make sure that $L$ is fully uncorrupted.

The main idea to recover $L_0$ and identify the set of corrupted challengers is to solve a convex program to find an optimal pair $(L^*,C^*)$ minimizing the weighted sum of the nuclear norm of $L$ (the sum of singular values of $L$) and the matrix $\ell_{1,2}$ norm of $C$ (column $\ell_{2}$ norm). The conditions of successful completion require (1) enough observed entries ( $p\gtrsim\frac{r\log^{2}m}{m}$) and (2) limited number of corrupted columns ($\beta\lesssim\frac{p}{r\sqrt{r}\log^{3}m}$). We implement Algorithm 2 in~\cite{chen2015matrix} which applies the Augmented Lagrangian Multiplier (ALM) to the solve the optimization problem. We evaluate the matrix completion results on a $100\times 100$ RTT matrix ($m=100$) with different faction of $\beta \in [0.1, 0.2, 0.3, 0.4]$ columns and different sample probability $p \in [0.3, 0.4, 0.5, 0.6]$. Other parameters are chosen according to the conditions that ensures high completion successful probability.  Figure~\ref{fig:rmc} demonstrates that more sampling and less corruptions improve on detection accuracy. Figure~\ref{fig:detection} gives an example output of the detection results in comparison with the actual corruptions.

\begin{figure}[h]
    \centering
    \includegraphics[width=0.45\textwidth]{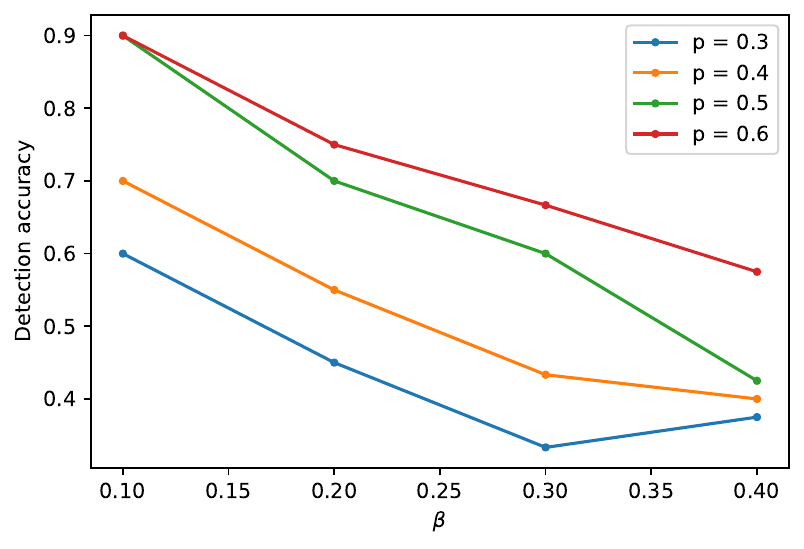}
    \caption{Detection accuracy of robust matrix completion under different $\beta$ and $p$.}
    \label{fig:rmc}
\end{figure}

\begin{figure}[h]
    \centering
    \includegraphics[width=0.45\textwidth]{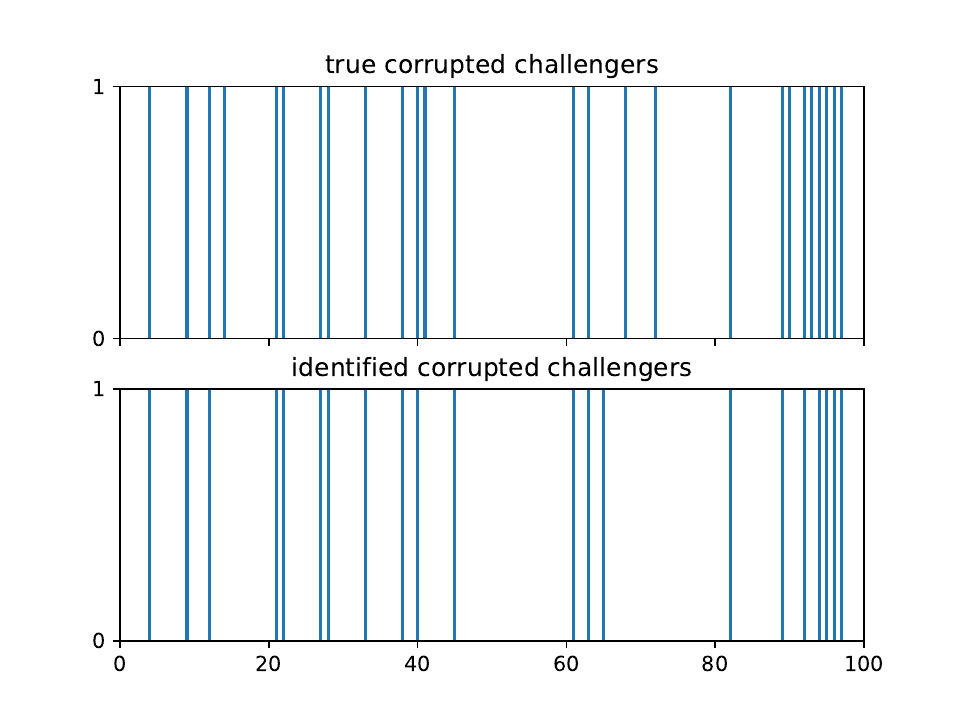}
    \caption{Detection results of robust matrix completion with $\beta=0.3$ and $p=0.6$.}
    \label{fig:detection}
\end{figure}

%% file: 06_evaluation.tex
\section{Implementation and Evaluation}
\label{sec:evaluation}

In this section, we describe the implementation of our protocol and discuss the measurement results in detail. The source code of our system is available at~\cite{github-link}.

\subsection{Protocol Implementation}
\label{subsec:protocol_implementation}
We implement the challengers and Waldo components in the PoLoc protocol as a Dart~\cite{dart} program. We selected Dart as the programming language for its robust support for multiple platforms~\cite{dart-multiple}, allowing crowdsourced challengers to run the program on their preferred platforms.

\begin{figure}[h]
    \centering
    \includegraphics[width=0.4\textwidth]{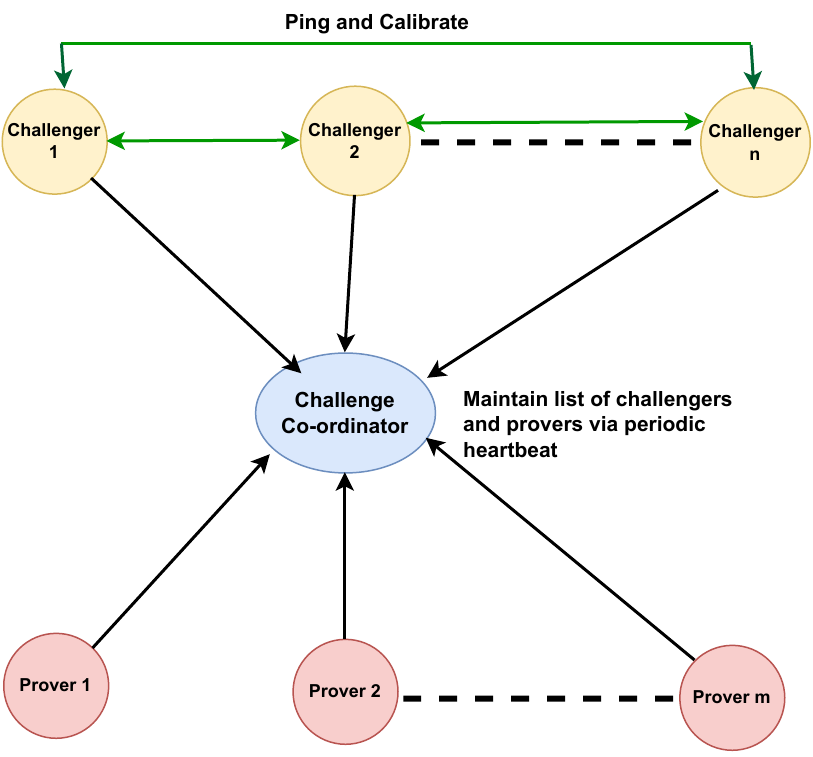}
    \caption{Implementation architecture with challenge coordinator.}
    \label{fig:implementation}
\end{figure}

Figure~\ref{fig:implementation} illustrates the implementation architecture. A microservice, named {\em challenge coordinator}, has been implemented in Python at the backend, who maintains a list of all active challengers. The program of challengers and Waldo continuously send heartbeats to the challenge coordinator, indicating their active status. The challengers periodically conduct the trusted PoIG protocol, measuring delays among themselves to construct a monotone delay-distance mapping, as mentioned in Section~\ref{sec:PoIG}. After mapping construction, each challenger commits this information to the challenge coordinator, who then uses the mappings to compute distances and output the final uncertainty. 

We deploy our BFT-PoLoc system on top of the Solana and Ethereum blockchains for challenge coordination and proof generation.
To initiate a challenge, we develop a user interface accessible via a web browser~\cite{pol-ui}. A payer with a blockchain account can initiate a challenge through the UI by paying the appropriate amount of tokens. Once initiated, the challenge coordinator informs the selected challengers and Waldo about the challenge. All challengers within $X=2000$ km of Waldo's claimed location are selected for the challenge.

The selected challengers initiate communication with Waldo. The challengers and Waldo authenticate each other using their respective public-private key pairs from the underlying blockchain the challenge was initiated on. After authentication, the challenger measures the delay to Waldo using 20 packets exchanged over UDP. We measure delays over UDP to avoid the overhead associated with TCP acknowledgments.

After measuring the delay, the challengers report it to the challenge coordinator. If a challenger's delay to Waldo exceeds the maximum delay measured during the PoIG phase, as outlined in Section~\ref{sec:PoIG}, that challenger is not considered for validating Waldo. The challenge coordinator has the delay-distance mapping for each challenger. Using this mapping, it can compute the final $R^{*}$ for Waldo following Algorithm~\ref{alg:PoLoc}.

\subsection{Evaluation Results}
\paragraph{Experimental setup.} We use 450 RIPE Atlas~\cite{ripe-atlas} nodes to simulate different Waldos. And we deploy our own set of challengers in public clouds~\cite{aws, oneprovider, vpsserver, stackpath}. Since RIPE Atlas nodes cannot run the custom program, we measure ICMP ping delays from challengers to RIPE Atlas Waldos instead. The challengers deployed in public clouds run the custom challenger program. We focus specifically on the task of determining whether a Waldo is in the US. Therefore, all challengers and Waldos are distributed in the US and neighboring countries. Figure~\ref{fig:challengers37} and ~\ref{fig:Waldos} show the distribution of these nodes, 34 challengers are deployed in the US, 2 are in Canada, and 1 is in Mexico. We investigate how the uncertainty $R^{*}$ varies when certain system parameters change, these parameters includes: (1) the number of challengers, (2) whether Waldo is Byzantine, (3) whether Waldos uses VPNs, and (4) the existence of Byzantine challengers. The results are discussed below.

\paragraph{Number of challengers.} To examine the effect of the number of challengers, we use $22$ and $37$ challengers to evaluate the uncertainty of Waldos. To test the performance of larger number of challengers, we further utilize $450$ RIPE Atlas nodes as challengers. We consider two different situations: (1) honest Waldos located in the US proving their locations, (2) Byzantine Waldos outside the US claiming a location within the US. 

As is shown in Figure~\ref{fig:cdf_uncertainty}, with 37 challengers in place, the system could verify the location of Waldos with significantly reduced uncertainty compared to using only 22 challengers. Specifically, when using 22 challengers about 20\% of the Waldos can be validated with uncertainty less than 100 kms, while using 37 challengers the fraction increases to about 25\%. 450 challengers can further improve the fraction to 45\%. Alternatively, the percentage of Waldos that achieves uncertainty less than 1000 km for 22, 37, 450 challengers are about 80\%, 85\% and 95\%, respectively.
This reduction in uncertainty directly correlates with the system's ability to authenticate the location of Waldos more precisely, showcasing the importance of a large and diverse set of challengers in enhancing the system's effectiveness.

These results demonstrate that as the number of challengers increases, the uncertainty in validation decreases. This is because an increase in the number of challengers leads to more challengers with different ISPs. The increase in ISP diversity makes the PoIG of the delay-distance mapping outlined in more robust, resulting in performance gains. Moreover, our results suggest that some Waldos, despite high challenger diversity, can still exhibit higher uncertainty if they have a poor internet connection to neighboring ISPs. However, the number of such Waldos is small, with only 5\% having an uncertainty of more than 1000 km with 450 challengers.

\begin{figure}
\centering
\begin{subfigure}[]{0.4\textwidth}
   \includegraphics[width=1\linewidth]{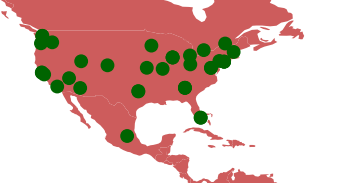}
   \caption[]{}
   \label{fig:challengers37} 
\end{subfigure}

\begin{subfigure}[]{0.4\textwidth}
   \includegraphics[width=1\linewidth]{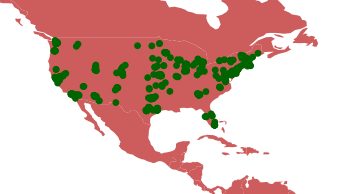}
   \caption[]{}
   \label{fig:Waldos}
\end{subfigure}

\caption[]{(a) 37 challengers deployed in different public clouds in US, Canada and Mexico. (b) 450 RIPE Atlas Waldos in US.}
\end{figure}

\begin{figure}[h]
    \centering
    \includegraphics[width=0.45\textwidth]{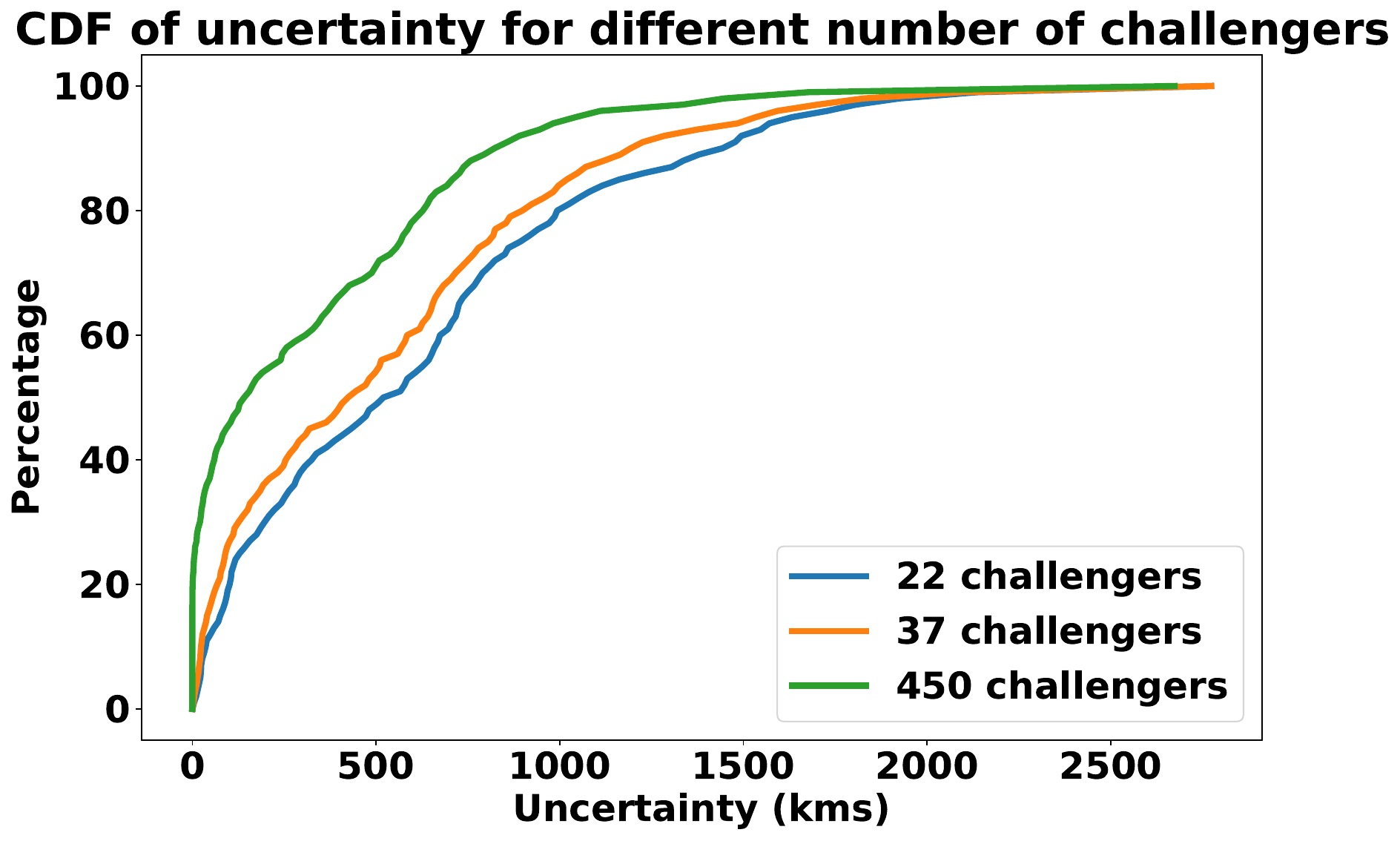}
    \caption{CDF of uncertainty for different number of challengers.}
    \label{fig:cdf_uncertainty}
\end{figure}

\paragraph{Byzantine Waldos.}  A critical aspect of our evaluation is assessing the system's robustness against Byzantine behaviors, including Waldos attempting to falsify their location and challengers providing misleading measurements. Our experiments with Byzantine Waldos, who claimed a location within the US but were actually positioned outside, highlighted the system's capacity to identify such discrepancies. We use an uncertainty threshold to effectively distinguish honest Waldos from Byzantine ones, with a trade-off between the accuracy of detection and the potential for false positives among honest Waldos.

\begin{figure}
\centering
\begin{subfigure}[]{0.4\textwidth}
   \includegraphics[width=1\linewidth]{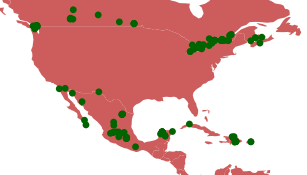}
   \caption[]{}
   \label{fig:byzantine_Waldos_actuallocation} 
\end{subfigure}

\begin{subfigure}[]{0.4\textwidth}
   \includegraphics[width=1\linewidth]{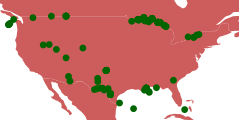}
   \caption[]{}
   \label{fig:byzantine_Waldos_claimedlocation}
\end{subfigure}
\caption[]{Byzantine Waldos ouside US (a) True location outside US. (b) Claimed location in US.}
\end{figure}

We chose 206 RIPE Atlas nodes serving as Waldos outside US in Mexico, Canada, Puerto Rico and Cuba. Fig.~\ref{fig:byzantine_Waldos_actuallocation} shows the true location of the byzantine Waldos, while Fig.~\ref{fig:byzantine_Waldos_claimedlocation} shows the claimed location of the Waldos in the US. We measure the uncertainty of their claimed location in the US with different number of challengers. 

\begin{figure}[h]
    \centering
    \includegraphics[width=0.4\textwidth]{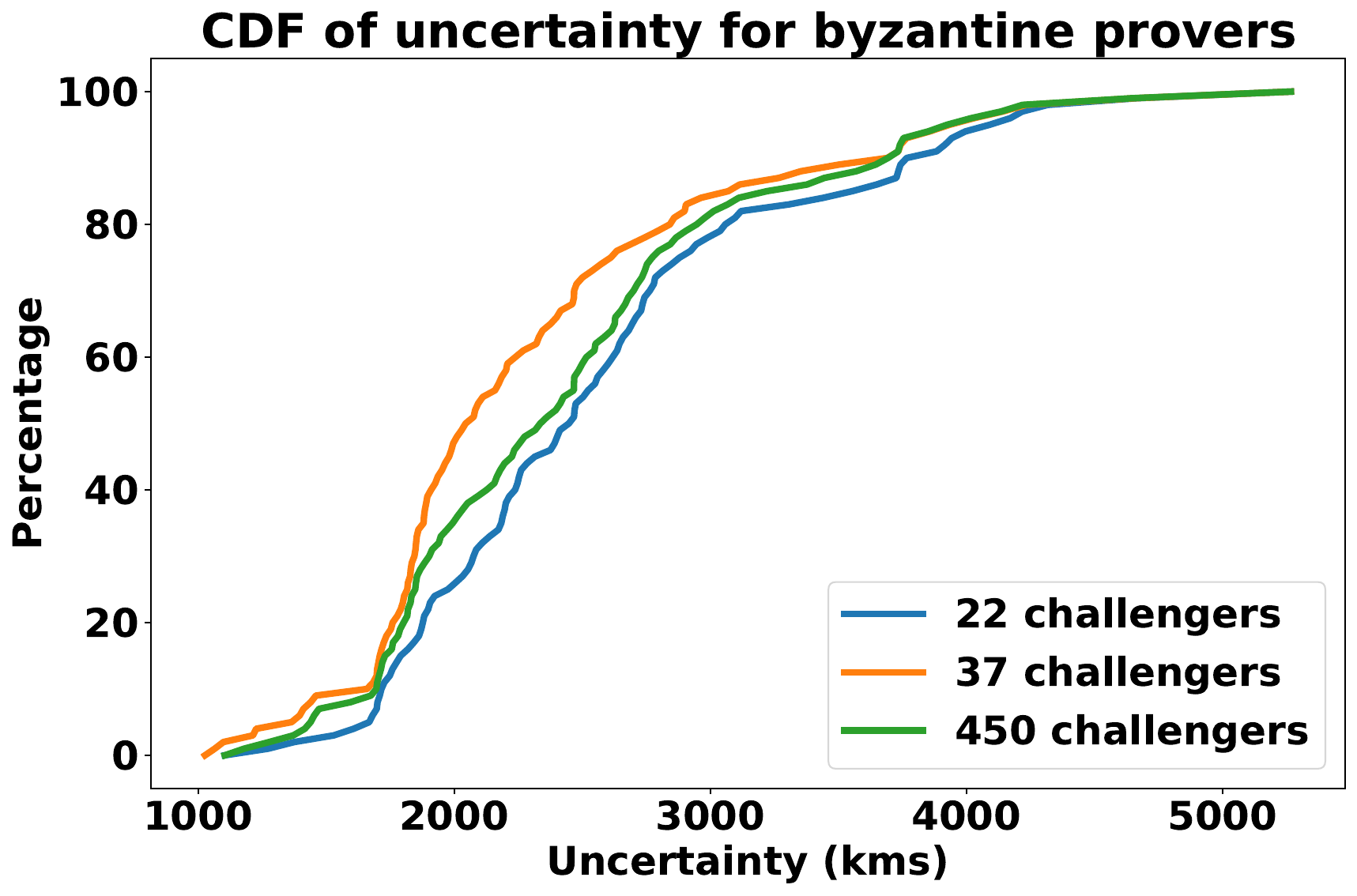}
    \caption{CDF of uncertainty for byzantine Waldos outside US.}
    \label{fig:cdf_uncertainty_byzantine}
\end{figure}

Figure~\ref{fig:cdf_uncertainty_byzantine} illustrates the CDF of uncertainty for the examined Byzantine Waldos. Notably, the level of uncertainty exhibits minor variations across three distinct sets of challengers. The results evaluated by 37 challengers display marginally lower uncertainty. In each scenario, the minimum uncertainty exceeded 1000 kilometers. Consequently, employing a 1000-kilometer threshold as the criterion for distinguishing Byzantine Waldos enables the accurate identification of such false location claims. Nonetheless, adopting this threshold comes with a caveat; as depicted in Figure~\ref{fig:cdf_uncertainty}, applying it results in a misclassification rate where 5\% of honest Waldos within the US, are inaccurately deemed to be outside when utilizing 450 challengers. Similarly, this misclassification rate escalates to 15\% with 37 challengers and further to 20\% with 22 challengers. We can reduce the misclassification rate by increasing the uncertainty threshold, but not all Byzantine Waldos can be detected. For example, with a threshold of 1500 km, we can identify 90\% of byzantine Waldos, and reduce the rate of misclassfication of valid Waldos to 2\%, with number of challengers as 450.  It demonstrates the trade-off in balancing false positives against the accurate detection of Byzantine behaviors.

\paragraph{The use of VPNs.}
We conducted a series of experiments to assess the impact of using a VPN server located in the US on the uncertainty associated with Byzantine Waldos situated outside the US. The protocol measures the end-to-end delay between the challenger and Waldo, ensuring authenticity through public-private key validation. This measurement approach implies that employing a VPN server would introduce additional delay beyond the natural latency between the challenger and Waldo.

\begin{figure}[h]
    \centering
    \includegraphics[width=0.5\textwidth]{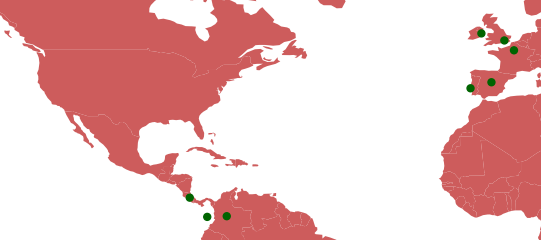}
    \caption{Byzantine Waldos using VPN in USA.}
    \label{fig:vpn_Waldos}
\end{figure}

To validate this, we carried out experiments with 8 Waldos, with 5 of the Waldos in Europe and 3 in South America as shown in Fig.~\ref{fig:vpn_Waldos}. The European Waldos were using a VPN in Virginia, in the east coast of US which is closer to them, while the South American Waldos were using a VPN in Texas, which is closer to them.

To investigate this effect, we analyzed the performance of 8 Waldos: 5 in Europe and 3 in South America, utilizing VPN servers in Virginia and Texas, respectively, due to their geographical proximity. Figure~\ref{fig:vpn_Waldos} depicts the CDF of uncertainty for these Waldos, contrasting scenarios with and without VPN usage. Notably, two South American Waldos were disqualified as their delays exceeded the benchmarks established during the PoIG phase, automatically categorizing them as Byzantine under our protocol. The analysis revealed that Waldos not using VPNs exhibited a minimum uncertainty of approximately 5000 km, whereas VPN usage inflated this uncertainty to over 9000 km in the worst-case scenarios.

These findings highlight a significant disparity compared to the 1000 or 1500 km uncertainty thresholds used for Byzantine Waldo identification. Consequently, our protocol demonstrates robustness against attempts to falsify location via VPN, as these strategies markedly increase the uncertainty, facilitating the identification of Byzantine Waldos located outside the US continent. In essence, Waldos attempting to manipulate their location through VPN inadvertently compromise their deception by exacerbating the detectable delay discrepancies, thereby reinforcing the protocol's efficacy in distinguishing malicious behaviors.

\begin{figure}[h]
    \centering
    \includegraphics[width=0.38\textwidth]{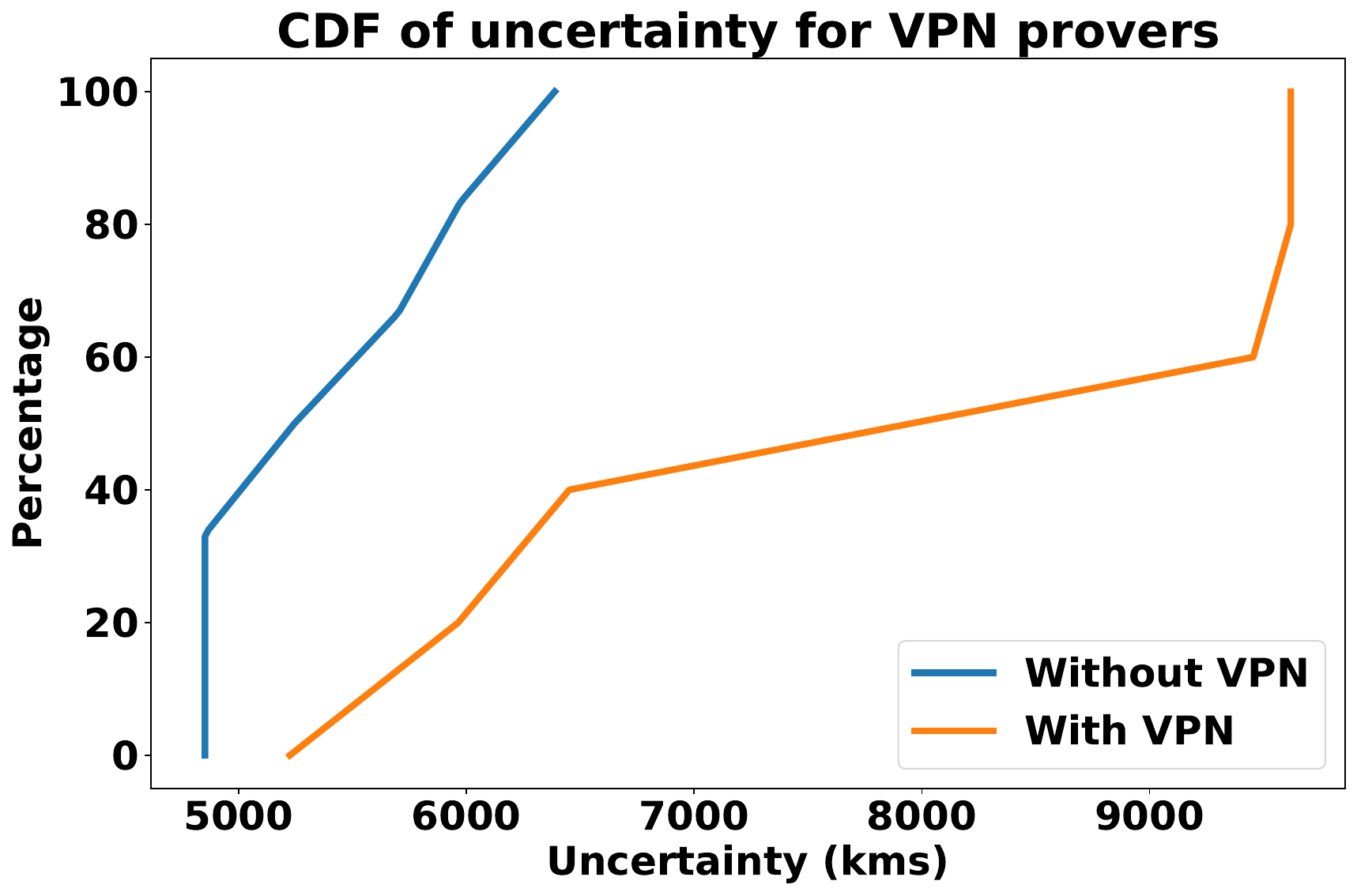}
    \caption{CDF of uncertainty for byzantine Waldos with and without VPN.}
    \label{fig:cdf_vpn_Waldos}
\end{figure}

\paragraph{Byzantine challengers.}
\label{subsubsec:byzantine_challengers}

Our evaluation so far assumes all challengers are trustworthy. However, the introduction of Byzantine challengers, who may either inflate or deflate the delay to manipulate Waldo's perceived location, necessitates an evaluation of our protocol's resilience to such adversarial behavior.

A method to mitigate the impact of Byzantine challengers involves disregarding the most extreme uncertainty measurements, effectively eliminating data points that could skew the results due to malicious intent. To quantify the protocol's tolerance for Byzantine behavior, we analyzed the effect of excluding the top $q$ percentile of uncertainty measurements, exploring the system's performance across varying levels of adversary presence.

Figure~\ref{fig:cdf_byzantinechallengers} illustrates the uncertainty distribution for a scenario with 450 challengers, including cases with no Byzantine challengers and with 2\%, 5\%, and 10\% Byzantine participation. By excluding the highest 2\% of uncertainty values for the scenario with 2\% Byzantine challengers, and similarly for 5\% and 10\%, we observed a notable differentiation in the system's ability to accurately identify Waldo's location. In the absence of Byzantine challengers, 95\% of Waldos were correctly validated within a 1000 km uncertainty threshold. This accuracy decreases as the percentage of Byzantine challengers increases, highlighting the need for a substantial pool of challengers to dilute the effect of malicious actors effectively.

\begin{figure}[h]
    \centering
    \includegraphics[width=0.4\textwidth]{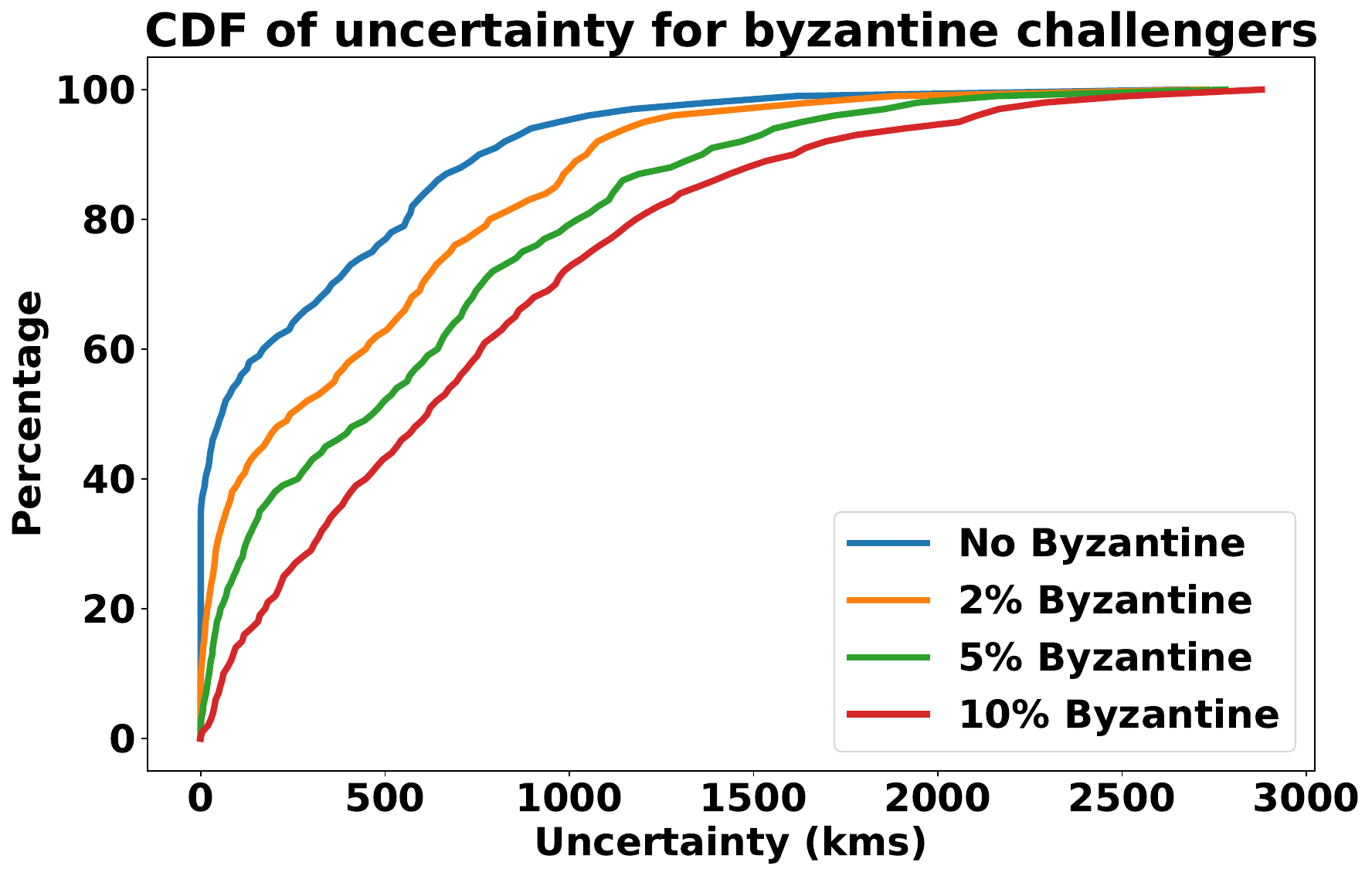}
    \caption{CDF of uncertainty with byzantine challengers.}
    \label{fig:cdf_byzantinechallengers}
\end{figure}
Furthermore, we explored the potential for enhanced tolerance to Byzantine challengers by adjusting the protocol's angular coverage requirements. Specifically, we examined scenarios where uncertainties within a 60-degree angular region around the maximum uncertainty angle were disregarded, effectively focusing on a 300-degree coverage area.

The modified approach, as depicted in Figure~\ref{fig:cdf_byzantinechallengers2}, revealed an improved ability to maintain accuracy despite Byzantine interference. For instance, excluding a 60-degree angular region increased the percentage of Waldos correctly identified within the 1000 km threshold by approximately 2\% to 6\%, depending on the percentage of Byzantine challengers present. This adjustment demonstrates that by relaxing the coverage requirements slightly, the system can afford a higher tolerance for Byzantine behavior, potentially accommodating up to 10\% Byzantine challengers with an acceptable level of accuracy.

In summary, our findings suggest that while the presence of Byzantine challengers poses a significant challenge, strategic adjustments to the protocol's evaluation criteria can enhance resilience, enabling it to withstand a higher proportion of adversarial participation without compromising the integrity of location verification.

\begin{figure}[h]
    \centering
    \includegraphics[width=0.4\textwidth]{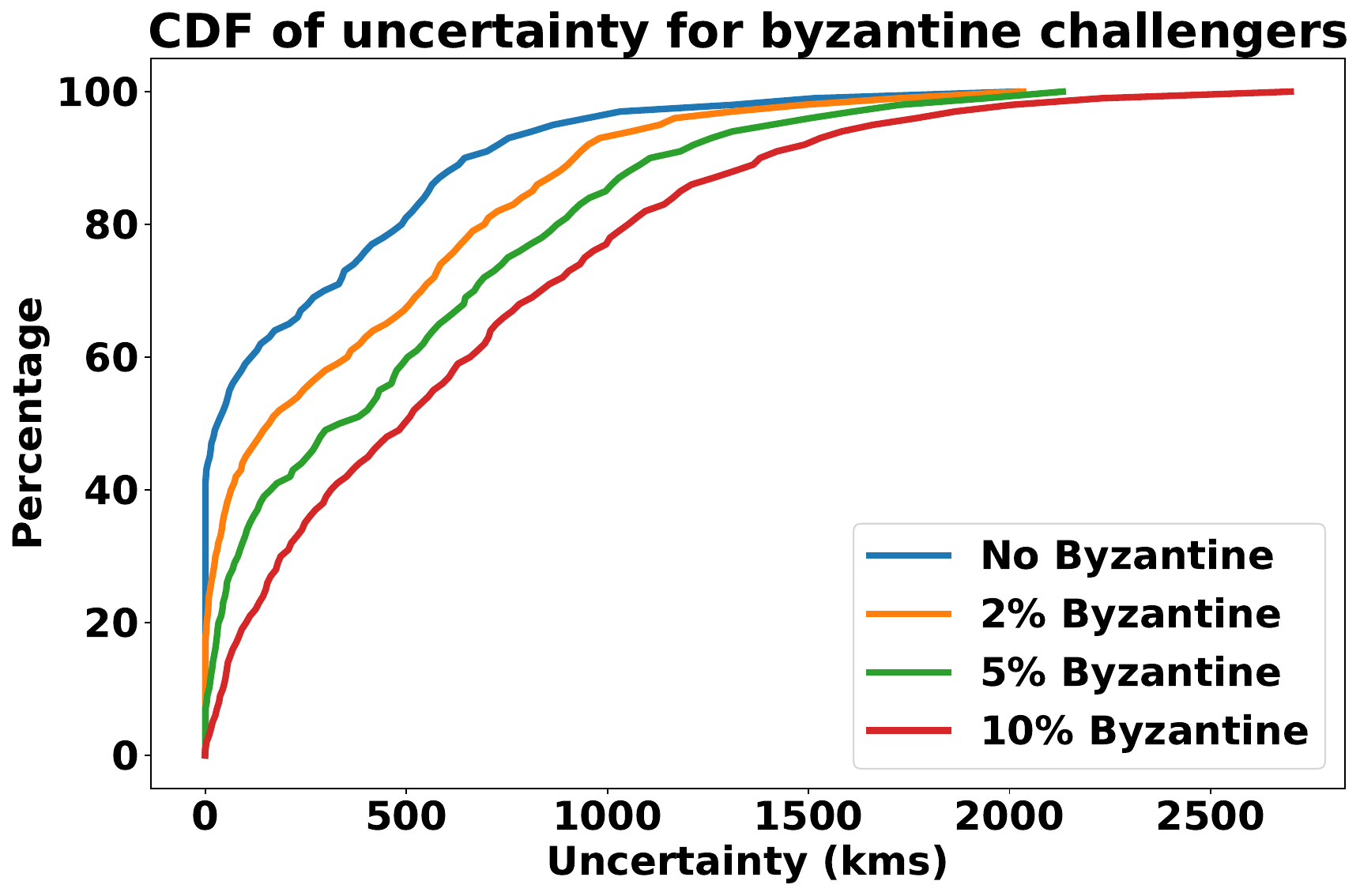}
    \caption{CDF of uncertainty with byzantine challengers with 60 degree angular region uncovered.}
    \label{fig:cdf_byzantinechallengers2}
\end{figure}

\subsection{Comparison with CBG} 
\label{subsubsec:cbg}
In this analysis, we explore whether a monotone mapping approach enhances performance compared to a linear mapping method, exemplified by the CBG model~\cite{gueye2004constraint}, which employs a best-fit line to correlate delay with distance. In the CBG model, a given delay is translated into the maximum possible distance between the Waldo and a challenger through linear regression.

\begin{figure}[h]
    \centering
    \includegraphics[width=0.4\textwidth]{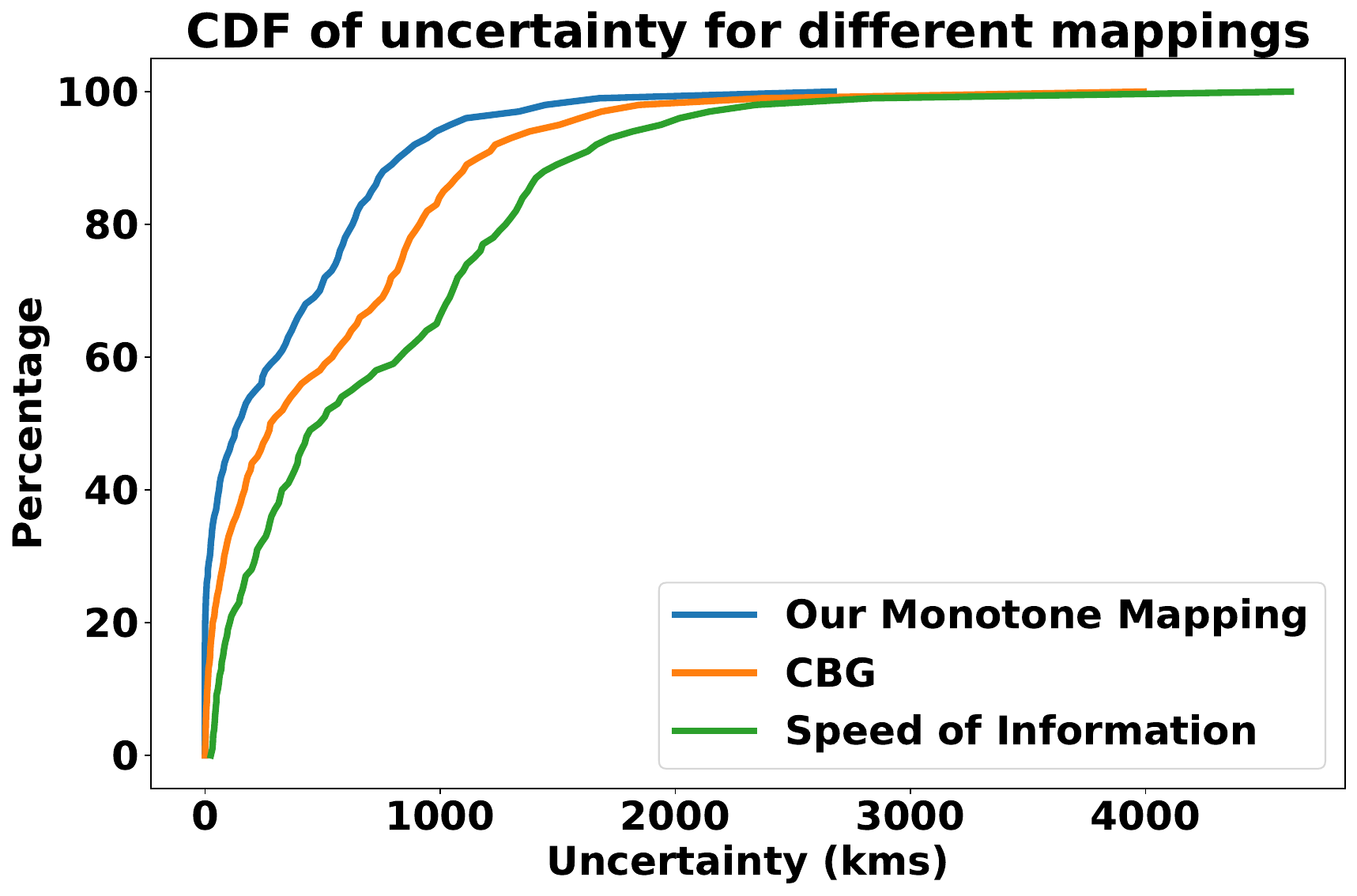}
    \caption{CDF of uncertainty for different delay to distance mappings.}
    \label{fig:cdf_mappings}
\end{figure}

Figure~\ref{fig:cdf_mappings} compares the uncertainty distributions for our monotone mapping strategy, the CBG approach, and the theoretical uncertainty based on the speed of information transmission in fiber optic cables, represented as $\frac{2}{3}c$ or approximately 200 km/s, where $c$ is the speed of light. The analysis is conducted with a scenario involving 450 challengers.

The data reveal that, at the 95th percentile, the uncertainty threshold of 1000 km—which we adopt for identifying Byzantine Waldos—is met by our monotone mapping approach, whereas the CBG method reaches this threshold at approximately the 85th percentile, and the theoretical model based on the speed of information in fiber achieves it at merely the 67th percentile. This disparity indicates that the CBG model, with its linear mapping, tends to overestimate distances, thereby increasing the uncertainty in comparison to our monotone mapping method.

 The comparative analysis highlights the potential for monotone mapping to refine the accuracy of delay-to-distance functions, reducing the margin of error and bolstering the system's overall performance in identifying location spoofing attempts.

%% file: 07_conclusion.tex
\section{Conclusion}
\label{sec:conclusion}
We develop a proof system to validate the geographical locations of Internet participants in a trustless environment, without relying on any trusted service. The system comprises two parts: a Proof of Internet Geometry protocol that robustly converts delay measurements into distance estimates, and a Proof of Location protocol which leverages signed ping delays and a Byzantine-fortified trigonometry framework to limit the uncertainty in geolocation. Essentially, our study utilizes internet delay measurements combined with cryptographic techniques and Byzantine-resistant data analysis within geometric constraints to securely verify locations.

We implement a fully functional location verification system and integrate it into major blockchains. Through comprehensive testing on distributed nodes across the US, we evaluate the system's accuracy and robustness and examine how various system parameters affect them. Our findings reveal that the protocol effectively identifies Byzantine Waldos claiming a U.S. location while actually being outside the country, with an uncertainty margin of 1000 km. The results further indicate that Byzantine Waldos attempting to masquerade as being in the U.S. by using VPNs from another continent are at a disadvantage. The protocol withstands up to 2\% Byzantine challengers while maintaining a 12\% error rate in the classification of Byzantine Waldos. Additionally, we demonstrate that our approach of monotonically mapping delay to distance yields considerable improvements over traditional linear mapping methods.